\documentclass[10pt, conference]{IEEEtran}

\usepackage{braket}
\usepackage{amsmath,amssymb,amsfonts,dsfont,amsthm}
\usepackage{algorithmic}
\usepackage{graphicx}
\usepackage{textcomp}
\usepackage{xcolor}
\usepackage{comment}
\usepackage{cite}
\usepackage{authblk}
\usepackage{xfrac}
\usepackage{tcolorbox}
\usepackage{xcolor}

\newcommand\blfootnote[1]{%
  \begingroup
  \renewcommand\thefootnote{}\footnote{#1}%
  \addtocounter{footnote}{-1}%
  \endgroup
}

\newtheorem{theorem}{Theorem}
\newtheorem{corollary}[theorem]{Corollary}
\newtheorem{remark}[theorem]{Remark}

\newtheorem{lemma}[theorem]{Lemma}
\newtheorem{definition}[theorem]{Definition}
\newtheorem{proposition}[theorem]{Proposition}

\newlength{\blank}
\newcommand{\1}{\mathds{1}}

\newcommand{\EE}{\mathbb{E}}
\newcommand{\RR}{\mathbb{R}}

\newcommand{\cB}{\mathcal{B}}

\newcommand{\cD}{\mathcal{D}}

\newcommand{\cF}{\mathcal{F}}
\newcommand{\cG}{\mathcal{G}}
\newcommand{\cH}{\mathcal{H}}

\newcommand{\cN}{\mathcal{N}}

\newcommand{\cT}{\mathcal{T}}

\newcommand{\cV}{\mathcal{V}}
\newcommand{\cX}{\mathcal{X}}
\newcommand{\cY}{\mathcal{Y}}

\newcommand{\cW}{\mathcal{W}}

\begin{document}

\title{Commitment over Gaussian channels with deterministic identification codes\protect\\[2mm] \textsc{\normalsize Dedicated to the memory of Anderson C. A. Nascimento}}

\author[1]{Pau Colomer\textsuperscript{$\,\star$,}}
\author[2,3]{Christian Deppe\textsuperscript{$\,\dagger$,}}
\author[1,3,4,5]{Holger Boche\textsuperscript{$\,\ddagger$,}}
\author[6,7]{Andreas Winter\textsuperscript{$\,\S$,}}

\affil[1]{\small\it Chair of Theoretical Information Technology, School of Computation, Information and Technology, \protect\\ Technische Universit\"at M\"unchen, Theresienstra{\ss}e 90, 80333 M\"unchen, Germany.\vspace{2pt}}

\affil[2]{\small\it Institute for Communications Technology, Faculty of Electrical Engineering, Information Technology, \protect\\and Physics, Technische Universit\"at Braunschweig, Schleinitzstraße 22, 38106 Braunschweig, Germany.\vspace{2pt}}

\affil[3]{\small\it BMFTR Research Hub 6G-life, Schleinitzstra{\ss}e 22, 38106 Braunschweig, and Arcisstra{\ss}e 21, 80333 M\"unchen, Germany.\vspace{2pt}}

\affil[4]{\small\it Cluster of Excellence, CeTI, Technische Universität Dresden, Georg-Schumann-Stra{\ss}e 9, 01187 Dresden, Germany.\vspace{2pt}}

\affil[5]{\small\it Munich Quantum Valley, Walther-von-Dyck-Straße 6, 85748 Garching, Germany; and\protect\\ Munich Center for Quantum Science and Technology Schellingstraße 4, 80799 München, Germany.\vspace{2pt}}

\affil[6]{\small\it Chair of Quantum Information and Computing, Department Mathematik/Informatik \protect\\ Abteilung Informatik, Universit\"at zu K\"oln, Weyertal 121, 50931 Köln, Germany.\vspace{2pt}}

\affil[7]{\small\it ICREA \&{} Grup d'Informaci\'o Qu\`antica, Departament de F\'isica, \protect\\ Universitat Aut\`onoma de Barcelona, 08193 Bellaterra (Barcelona), Spain.}

\date{2 June 2026} 

\maketitle

\pagestyle{plain}

\begin{abstract}
We establish a previously unnoticed connection between commitment (BC), a fundamental cryptographic primitive, and deterministic identification (DI), a post-Shannon communication setting. Specifically, for additive white Gaussian noise channels $\mathcal{G}$ we show that a reliable commitment protocol can be obtained from deterministic identification codes in both the honest but curious and fully dishonest security settings. This viewpoint allows BC schemes to inherit the asymptotic performance of DI codes. Indeed, we show that commitment is naturally achievable in the linearithmic regime, i.e., with message sets of size $N_n=\exp [\Theta(n\log n)]$, and we establish a lower bound on the linearithmic-scale BC capacity of $\dot C_{\text{hBC}}(\mathcal G)\geq\frac12$ in the honest but curious case and $\dot C_{\text{BC}}(\mathcal G)\geq\frac14$ in the fully dishonest picture. This refines the best previously known characterisation of the commitment capacity over $\mathcal{G}$, which established only an infinite linear-scale capacity. 
\end{abstract}


\blfootnote{
{\textsuperscript{$\star$}pau.colomer@tum.de, \hspace{20pt} \textsuperscript{$\dagger$}christian.deppe@tu-bs.de,}\vspace{1pt}

{\textsuperscript{$\ddagger$}boche@tum.de, \hspace{45pt}\textsuperscript{$\S$}andreas.winter@uni-koeln.de.}}

\section{Introduction}\label{sec:intro}
Cryptography was initially developed in the setting of secure communication between legitimate users in the presence of an external adversary. The classical problems of encryption, authentication, and key agreement aim to protect confidentiality and integrity against an eavesdropper or an active attacker controlling the communication channel. This paradigm has led to a rich mathematical theory, beginning with Shannon's foundational work on secrecy systems \cite{Shannon:secrecy} and evolving into modern computational cryptography.

A different class of cryptographic problems emerges when the communicating parties do not trust one another. Rather than cooperating against an external adversary, each participant may attempt to gain an advantage by deviating from the prescribed protocol. Bit-string commitment (BC) is one of the fundamental primitives on these mutually distrustful settings. It allows one party to commit to a secret value while keeping it hidden from a second party (the \emph{commit phase}) until a later date, at which point the secret is revealed and verified (the \emph{reveal phase}). Protocols such as distributed and secure multi-party computation, zero knowledge proofs, coin tossing, or contract signing \cite{Y:secure-computation,ZCGLT:secure-computation, BC:zero-knowledge, Blum:coin-flipping, EGL:contract-signing} use BC as a building block. 

A secure commitment scheme must simultaneously satisfy the three following (and apparently conflicting) requirements: 

\vspace{-1pt}
\begin{itemize}
    \item \textbf{Concealment}: The receiver cannot learn the committed word until it is revealed by the sender.
    \item \textbf{Bindingness}: The sender cannot change the word once it has committed.
    \item \textbf{Soundness}: After the commitment is revealed, the receiver must be able to verify it.
\end{itemize}

The commitment of classical messages was initially studied using noiseless channels. In such cases, local BC schemes cannot be perfectly concealing and perfectly binding simultaneously \cite{DKS:OT_and_BC} (this impossibility extends also to noiseless quantum channels \cite{Mayers:no-go,LC:no-go}). However, imposing computational limitations on the sender and/or the receiver, computationally secure BC over noiseless channels can be achieved \cite{Naor:BC_pseudorandomness, HM:BC_computationally_secure}. Furthermore, if the two users are spatially separated, unconditionally secure BC protocols over noiseless channels, without computational or time constraints, can be designed leveraging relativistic effects \cite{Kent:relativistic}.

The use of noisy communication channels as cryptographic resources, introduced by Wyner in his seminal work on wiretap channels \cite{Wyner:wiretap}, marked a new chapter on the field. Indeed, BC over (noisy) binary symmetric channels (BSCs) was shown to be unconditionally secure \cite{Crepeau:crypto_protocols, DKS:OT_and_BC}. This result opened a promising new line of research, with a focus on how to generalise and optimise noise as a resource for BC. 
\medskip

Much like in Shannon's theory of channel capacities \cite{Shannon:TheoryCommunication}, one wishes to commit to a single message $i\in [M]$ from a set of $M$ elements through $n$ uses of a noisy channel. The commitment \emph{rate} of a protocol is thus defined as $R(n)\!=\!\frac{1}{n}\log M$, and a rate $R_{\text{BC}}$ is said to be \emph{achievable} if there exist BC protocols for every $n$ with rates $R(n)$ converging to $R_{\text{BC}}$ and vanishing probability of errors as $n\to\infty$. The supremum of all achievable rates is known as the \emph{capacity of commitment} and denoted here by $C_{\text{BC}}$.
The BC capacity over discrete memoryless channels (DMCs), which include the previously studied BSCs, was established in \cite{WNI:BC_capacity_DMC}. 

\begin{theorem}[Commitment capacity for DMCs \cite{WNI:BC_capacity_DMC}]\label{thm:BC_DMC}
The commitment capacity of a non-redundant DMC $\cV$ is the maximal equivocation of the channel over all input distributions:
\[
C_{\text{BC}}(\cV)=\max_{p_X} H(X|Y),
\]
where $H(X|Y)$ denotes the conditional entropy, the average uncertainty about the transmitted input given the output.
\end{theorem}

The result above shows that unconditionally secure BC can be achieved in general by communicating through any non-redundant noisy DMC. Notice that, if the channel is noiseless, the conditional entropy vanishes and, hence, the capacity is zero, consistently with the impossibility of doing reliable commitment over noiseless channels. 

An efficient protocol to achieve the capacity in Theorem \ref{thm:BC_DMC} was introduced in \cite{IMNW:efficient_protocol}, and \cite{MYMB:cost-constraint} provided an extension to a setting with cost constraints. It was shown in \cite{Hayashi:commitment_list-decoding} that commitment could be achieved with a list-decoding strategy.

An interesting behaviour was observed when moving beyond discrete channels. Specifically, for additive white Gaussian noise (AWGN) channels $\cG$, which have continuous input and output alphabets, it was shown in \cite{NBSH:commitment_Gauss} that unconditionally secure commitment was possible and the commitment capacity $C_{\text{BC}}(\cG)=\infty$ is infinite! 

Since the rates have been defined as $R(n)=\frac{1}{n}\log M$, an infinite capacity means that the length $\log M$ of the messages scales faster than linearly with $n$ (making the capacity diverge as $n\to\infty$). But how fast does it really scale? Can we establish a meaningful value (not infinite or zero) for the BC capacity over AWGN channels defined in the correct super-linear scale? These questions motivate the present work.

Commitment over continuous channels was also studied in \cite{Hayashi:commitment_list-decoding}, but the rates there were deliberately kept on the linear scale to enforce the connection to list-decoding, hence losing the increased scaling behaviour.
\medskip

The difference in size scaling of optimal BC codes over discrete and continuous channels resembles that observed in deterministic identification. The identification of messages (ID) is a point-to-point communication paradigm in which, instead of reconstructing the sent message, the receiver is only concerned in reliably identifying if that sent message coincides with a specific one of its choosing, simplifying the decoding task to a hypothesis test \cite{Ja:ID_easier}. 

When randomness is available upon encoding, one can identify exponentially more messages than one could transmit \cite{AD:ID_ViaChannels, HanVerdu:ID}. In other words, while the length of the messages that can be reliably transmitted scales linearly with the block length $\log M\sim nR_\text{T}$, the length of the $N$ messages that can be reliably identified scales exponentially fast $\log N\sim 2^{nR_\text{ID}}$. The ID rates must thus be defined as $\ddot R_{\text{ID}}(n)=\frac{1}{n}\log\log N$ (where the double dot is used to indicate the double logarithmic scale). Otherwise, they would diverge as $n\to\infty$. In other words, since the capacity is the supremum of all achievable rates, an identification capacity defined in a linear-scale regime (like that of transmission) immediately produces $C_\text{ID}=\infty$. However, if the identification capacity is suitably defined in the exponential scale the following result is obtained.

\begin{theorem}[Randomised identification capacity \cite{AD:ID_ViaChannels,HanVerdu:ID}]\label{thm:RID}
For any memoryless point-to-point channel $\cW$, the identification capacity defined in the exponential scale equals the Shannon transmission capacity defined in the linear scale:
\[
\ddot C_\text{ID}(\cW)=\max_{p_X}I(X;Y)=C_\text{T}(\cW),
\]
where $I(X;Y)$ denotes the mutual information between the channel's input and output.
\end{theorem}

When randomness is not available to the encoder, a setting known as \emph{deterministic identification} (DI), the exponential gain of identification is lost \cite{AC:DI}. For DMCs, the length of the identifiable words can just be linear $\log N\sim nR_\text{DI}$, like in the typical Shannon transmission of messages setting, although with higher rates $R_\text{DI}\geq R_\text{T}$ \cite{SPBD:DI_power}.
\begin{theorem}[DI capacity for DMCs \cite{SPBD:DI_power}]\label{thm:DI_DMC}
Given an injective (a.k.a.~reduced) DMC $\cV'$ the DI capacity defined in the linear scale is 
\[
C_\text{DI}(\cV')=\max_{p_X}H(X)=\log[\#_\text{row}(\cV')]
\]
where $H(X)$ is the usual Shannon entropy of the input maximised over all input probability distributions, which is equal to the number of distinct rows of the reduced channel, denoted by $\#_\text{row}(\cV')$.
\end{theorem}

However, for channels with continuous input the scaling can be faster. In fact, the length of the identifiable messages without the use of randomness scales linearithmically $\log N\sim nR_\text{DI} \log n$. This effect was initially observed for Gaussian and Poisson channels \cite{DI-fading, RRGauss-arXiv, DI-poisson} and later generalised to any channel with continuous input and discrete output \cite{CDBW:Reliability-TCOM, CDBW:DI_classical}.

This new intermediate scaling requires, similarly to what we observed for unrestricted randomised identification, a suitable definition for the rates and capacity. The DI rates are defined as $\dot R(n)=\frac{1}{n\log n}\log N$, and the supremum of all achievable rates is the linearithmic scale DI capacity $\dot C_\text{DI}$ (with the single dot denoting this intermediate scaling). 
Notice that, since the DI code sizes for continuous channels $\cW$ satisfy $\log N\sim nR_\text{DI}\log n$, its DI capacity defined in the linear scale would be $C_\text{DI}(\cW)=\infty$, and the DI capacity defined in the exponential scale would be $\ddot C_\text{DI}(\cW)=0$. Only if the capacity is suitably defined in the linearithmic scale we obtain a meaningful positive and finite value describing the optimum rate (at the correct scale).
\begin{theorem}[DI capacity for continuous channels \cite{CDBW:DI_classical,CDBW:Gauss,CDBW:Bern_arxiv}]\label{thm:DI_continuous}
Let $\cW$ be a channel with continuous input alphabet $\cX$ and discrete output $\cY$, and let $\widetilde \cX=\cW(\cX)$ denote the input image. Then the DI capacity in the linearithmic scale satisfies
\[
\frac{1}{4}d_M\left(\sqrt{\widetilde\cX}\right)\leq \dot C_\text{DI}(\cW)\leq\frac{1}{2}d_M\left(\sqrt{\widetilde \cX}\right),
\]
where $d_M$ denotes the (lower) Minkowski dimension. For some channels with $d_M=1$ such as the Bernoulli or the Poisson channels, and some channels with continuous input \emph{and} output, such as the Gaussian channel, the capacity is exactly characterised by
\(
\dot C_\text{DI}(\cW)=\frac{1}{2}.
\)
\end{theorem}

In this paper, we will build and analyse BC codes for Gaussian channels inspired by a recent capacity-achieving DI code construction \cite{CDBW:EW_simplified, CDBW:Gauss, CDBW:Bern_arxiv}. An initial motivation to try this lies in the difference in the scaling behaviour between discrete and continuous channels that both BC and DI present. Indeed, for DMCs, we have observed that the length of the codes scales linearly in both cases (so the corresponding linear-scale capacity is finite, see Theorems~\ref{thm:BC_DMC} and \ref{thm:DI_DMC}); but, on the other hand, for Gaussian channels $\cG$ (which have continuous input) one observes an infinite value of the linear scale capacity for both commitment and deterministic identification $C_\text{BC}(\cG),C_\text{DI}(\cG)=\infty$ (see \cite{NBSH:commitment_Gauss} and the discussion before Theorem~\ref{thm:DI_continuous}). Additionally the DI capacity is characterised for Gaussian channels in its linearithmic scale $\dot C_\text{DI}(\cG)=\frac12$ (see Theorem~\ref{thm:DI_continuous}). We will show that a BC code can be created via a DI code, obtaining a linearithmic word length scaling, and we will lower bound the corresponding linearithmic BC capacity.

A perhaps more direct motivation stems from the similarity between the two protocols, which we discuss at this point only at an intuitive level. Let the sender and the receiver (Alice and Bob) dispose of a noiseless and a noisy channel $\cW$, and also agree on a deterministic identification code.

\begin{definition}\label{def:DI_code}
An $(n,N,\lambda_1,\lambda_2)$ code for DI over $n$ uses of a noisy channel $\cW:\cX\to\cY$ is a family $\{(u_i,D_i):i\in[N]\}$ of code words $u_i\in\cX^n$ and associated decoding regions $D_i\subset\cY^n$ satisfying, for all $i\neq j\in [N]$,
\[
\cW^n(D_i|u_i)\geq1-\lambda_1, \quad \cW^n(D_j|u_i)\leq\lambda_2.
\]
Here $\lambda_1$ denotes the maximum probability of missed identification (the receiver tests the word sent but gets a negative outcome on the identification test), and $\lambda_2$ bounds the probability of wrong identification (the words sent and tested are different, but the outcome of the identification test is positive). These $\lambda_1$ and $\lambda_2$ vanish as $n\to\infty$ in reliable DI codes.
\end{definition}

In the initial phase, Alice commits to a message $i\in[N]$ and sends the corresponding DI code word $u_i$ through the noisy channel. Later, in the reveal phase, Alice announces her commitment $i$ and Bob checks if the (noisy) sequence he received is in the associated identification region $D_i$. How far is this scheme from a good BC protocol?

First, notice that, after honest behaviour by both parties, the probability that Alice's commitment is rejected by Bob is exactly a missed identification (Bob tests the sent message but the decoder rejects it), which for any reliable DI code is bounded by a vanishing $\lambda_1$. Hence, the BC protocol trivially defined by a DI code is \emph{sound} up to a vanishing error probability $\lambda_1$. 

Second, during the commit phase, Bob receives the noisy output of a DI code word. He might try to guess the word that was sent before the reveal phase. It is a well-known fact in identification that if the receiver tests for all words, he will get, with high probability, multiple positive answers, due to the decoding regions in the identification setting having abundant overlaps with one another although pairwise small. This is the origin of wrong identification errors ($\lambda_2$ above) and, in fact, the starting point for the list-decoding constructions in \cite{Hayashi:commitment_list-decoding}. In other words, the receiver cannot in general learn the sent word reliably before the reveal phase. Commitment, however, requires a stronger notion of concealment: the information that the receiver is able to retrieve after the commit phase concerning the committed message must be small. Can we move from the non-decodability granted by DI codes to the strong concealment needed in BC?

Concerning bindingness, we distinguish two different security models known as \emph{honest but curious} and \emph{dishonest}. In the honest but curious case, the participants follow the exact rules of communication but might try to get an advantage from them. In the dishonest case, the participants do not need to follow the communication rules they agreed upon. In both these cases, the noisy channel parameters are fixed and cannot be manipulated by the users. The paradigms where one of the users might secretly tamper with the channel parameters to increase their cheating chances are known as \emph{unfair noisy settings} \cite{DKS:OT_and_BC,CDN:capacity_UNC,BJMY:noisy_settings_AWGN}.

In the honest but curious case, Alice might try to cheat by sending one code word $u_i$, and revealing a different one $u_j\neq u_i$ as her commitment. However, the probability that Bob accepts false commitment is bounded by the second type of error: $\cW^n(D_j|u_i)\leq\lambda_2$. Hence, the BC protocol trivially defined by a DI code is \emph{binding} up to a vanishing error probability $\lambda_2$.

In the dishonest case, Alice might send an input sequence (not from the code book) intending to cheat by revealing one of two different commitments which Bob might accept, breaking bindingness. This is something that is not contemplated in the DI setting, where the sender is assumed to honestly input code words to the noisy channel. Can we construct a \emph{binding} DI code, that is, a DI code such that any input sequence can only be reliably identified by at most one decoder? Can we then use this modified binding DI code to create a reliable BC code?



\section{Preliminaries}\label{sec:preliminaries}
In this paper we consider memoryless point-to-point communication over $n$ uses of an \emph{additive white Gaussian noise (AWGN)} channel. This channel acts on an input sequence $x^n=x_1,\dots,x_n\in\RR^n$ by adding a random Gaussian sequence $Z^n=Z_1,...,Z_n\in\RR^n$ under noise parameter $\sigma$, with each $Z_i\sim\cN(0,1)$ independently and normally distributed:
\begin{equation}\label{eq:channel_effect}
    Y^n=x^n+\sigma Z^n.
\end{equation}

In other words, letting $I_n$ be the $n\times n$ identity matrix, the channel output probability conditioned on the input sequence $x^n$ is distributed as follows:
\begin{equation}
    \cG_{x^n}:=\cG(\,\cdot\,|x^n)=\cN(x^n|\sigma^2\1).
\end{equation}

The AWGN channel occupies a central role in information theory and communications \cite{AWGN_IT, AWGN_radio, AWGN_wireless}. It serves as the canonical model for thermal noise in electronic systems and provides an accurate abstraction for a wide range of practical communication scenarios. Moreover, many more complex channel models reduce to Gaussian noise models in high-dimensional regimes or after appropriate processing.
Also, from a technical side, this channel has the advantage of the noise being isotropic and independent from the input word, which makes it easy to trace the (sometimes complex) needed properties of the output probability distributions directly to (usually simpler) geometrical properties of the corresponding input sequences.

Since practical scenarios cannot handle arbitrary large signals, Gaussian systems are mostly studied under a power constraint. Here we assume a block (or average) power constraint $P>0$ on the input sequences. In other words, each admissible $x^n$ must satisfy
\begin{equation}\label{eq:power_constraint}
\|x^n\|^2_2\leq nP,
\end{equation}
where $\|x^n\|_2 = \left(\sum_{i=1}^n x_i^2\right)^{1/2}$ is the Euclidean norm. Thus, notice that the available input space is an $n$-dimensional ball of radius $\sqrt{nP}$ centred at the origin $\vec o=(0,\dots,0)$
\begin{equation}
\mathcal{B}_{\vec o}(\sqrt{nP},n) 
:= \left\{ x^n \in \mathbb{R}^n : \|x^n\|_2 \leq \sqrt{nP} \right\}. 
\end{equation}

The Gaussian cumulative distribution function is denoted, for all $x\in\RR$ by
\begin{equation}\label{eq:comulative}
    \Phi(x)=\int_{-\infty}^x\frac{1}{\sqrt{2\pi}}e^{-z^2/2}dz.
\end{equation}

The \emph{total variation distance}, which describes a statistical distance between two (in this paper, continuous) probability distributions, say $P$ and $Q$ on the same continuous alphabet $\cY$, is denoted by
\begin{equation}\label{eq:TVD}
    \|P_Y-Q_Y\|_{\text{TV}}=\frac12\int_{y\in\cY}|P(y)-Q(y)|dy.
\end{equation}

The statistical distance above cannot increase under a stochastic (and therefore, either for a deterministic) mapping. Indeed, let
let $\Psi:\mathcal{Y}\rightarrow\mathcal{Z}$ be a deterministic mapping. Then,
\begin{equation}\label{eq:data-processing-TVD}
\|P_Y,Q_Y\|_{\text{TV}}
\geq \|P_{\Psi(Y)},Q_{\Psi(Y)}\|_{\text{TV}}.
\end{equation}

Let $X$ and $W$ be random variables with discrete alphabets $\mathcal{X}$ and $\mathcal{W}$, respectively. The Shannon entropy of $X$ is 
\begin{equation}
H(X)=-\sum_{x\in\mathcal{X}} p_X(x)\log p_X(x),
\end{equation}
where the $\log$ and $\exp$ are taken in base $2$ throughout the paper, and the sum should turn to an integral for continuous alphabets. The conditional Shannon entropy of $X$ given $W$ is:
\begin{equation}
\begin{split}
H(X|Y)
&= \sum_{w\in\cW} p_W(w)\, H(X| W=w)\\
&= -\sum_{x\in\cX}\sum_{w\in\cW}
p_{X,W}(x,w)\log p_{X|W}(X|w).
\end{split}
\end{equation}
In the context of communication, $H(X|W)$ is referred to as the \emph{equivocation} of the channel, and it has already appeared in Theorem~\ref{thm:BC_DMC} as the object ruling the BC capacity of DMCs.

The mutual information, which quantifies the amount of information shared by $X$ and $Y$, is denoted by
\begin{equation}\label{eq:mutual_information}
\begin{split}
I(X;Y)
&=\sum_{x\in\mathcal{X}}\sum_{y\in\mathcal{Y}}
p_{X,Y}(x,y)
\log\frac{p_{X,Y}(x,y)}{p_X(x)p_Y(y)}\\
&=H(X)-H(X|Y).
\end{split}
\end{equation}
It has already appeared in Theorem~\ref{thm:RID} as the quantity that rules the (linear-scale) capacity of transmission and the (exponential-scale) randomised identification capacity. Specifically, for an AWGN channel with variance $\sigma^2$ and average power constraint $P$, the capacity is \cite{Shannon:TheoryCommunication, Shannon:noise, Burnashev:ID_gauss, HanVerdu:ID}:
\begin{equation}
\!\!\!C_\text{T}(\cG)=\ddot C_\text{ID}(\cG)=\!\!\!\!\sup_{p_X:\EE X^2\leq P}\!\!\!\! I(X;Y)=\frac12\log\!\left(1+\frac{P}{\sigma^2}\right)\!.
\end{equation}
In contrast, and as established in Theorem~\ref{thm:DI_continuous}, the DI capacity in the linearithmic scale depends instead on a certain geometrical aspect of the output probability space and, for an AWGN channel, it is $\dot C_\text{DI}(\cG)=\frac12$ \cite{DI-fading,CDBW:Gauss}.

We will also use the \emph{conditional min-entropy} of $X$ (in a discrete alphabet $\cX$) given $Y$ (in a continuous alphabet $\cY$) with density function $p_Y$. It is defined as
\begin{align}
H_{\min}(X|Y)=-\log\int_{y\in\cY}p_{Y}(y)\max_{x\in\cX}p_{X|Y}(x|y)dy.
\end{align}

The vector in $\RR^n$ whose elements correspond to each term (letter) of the sequence $x^n$ is simply denoted $\vec x$. Working in this dual vector notation will sometimes alleviate the notation and help with the intuition of our very geometrical construction.
\medskip

The DI codes proposed in \cite{CDBW:Gauss,CDBW:Bern_arxiv}, a modification of which will constitute our construction for BC, crucially depend on projection properties between sequences. 
Let $\vec a,\vec v\in\RR^n$ with $\|\vec a\|>0$. If $\vec a$ is centred at the origin, we define the \emph{orthogonal projection} of $\vec v$ onto $\vec a$ is the linear map
\begin{equation}\label{eq:proj_definition}
\Pi_{\vec a}\,\vec v := \frac{(\vec a,\vec v)}{\|\vec a\|^2}\vec a,
\end{equation}
with $(\vec a,\vec v)=\sum_{i=1}^n a_i v_i$ the Euclidean inner product. 
%
%
The projection orthogonal of random Gaussian noise vectors onto fixed directions satisfies the following key property.

\begin{proposition}\label{prop:small_projection}
    Let $\vec{u}\in\RR^n$ be a vector in a high dimensional space and $\vec{Z}=(Z_1,\dots,Z_n)$ a random noise vector where $Z_1,\dots,Z_n\sim\cN(0,1)$ are i.i.d. standard normal random variables. The projection of the noise vector onto  $\vec{u}$ satisfies:
    \[
    \Pr\left\{\|\Pi_{\vec{u}}\vec{Z}\|_2\geq x\right\}=2\Phi(-x)\leq\sqrt{\frac{2}{\pi}}\frac{1}{x}e^{-x^2/2}.
    \]
\end{proposition}
\medskip

A hash function is a deterministic mapping from a large input domain to a shorter binary string. Since many inputs are mapped to a smaller output space, collisions are unavoidable; a $2$-universal family guarantees that, for any two distinct inputs, a randomly chosen hash function makes them collide with low probability. To be precise, let $\cH$ be a finite set and $r\geq 1$. A family of hash functions $\mathcal{F}\subseteq\{f:\cH\rightarrow\{0,1\}^r\}$ is called \emph{$2$-universal} if, for every pair of distinct inputs $h,h'\in\cH$,
\begin{equation}
\Pr\{f(h)=f(h')\}\leq 2^{-r},
\end{equation}
where $f$ is chosen uniformly at random from $\mathcal{F}$. A $2$-universal hash function can be used to extract an approximately uniform random string from a source with sufficiently large min-entropy, even when the hash function itself is revealed. This is formalized by the following lemma.

\begin{lemma}[Leftover Hash Lemma \cite{LHL_original,BBCM:privacy_ampl}]\label{lemma:LHL}
Let $P_{HB}$ be an arbitrary joint distribution over $\cH\times\cB$, denote by $F$ a uniformly random element of a $2$-universal family of hash functions $\mathcal{F}\subseteq\{f:\cH\rightarrow\{0,1\}^r\}$, and let $U_r$ be uniformly distributed over $\{0,1\}^r$ and independent of $(F,B)$.
Then,
\[
\left\|P_{F(H),F,B}-P_{U_r}P_{F,B}\right\|_{\text{TV}} \leq \frac{1}{2}2^{-\frac{1}{2}\left(H_{\min}(H|B)-r\right)}.
\]
\end{lemma}


We finalize the preliminaries section with the general definition of a commitment code over an AWGN channel $\cG$.

\begin{definition}\label{def:commitment}
An $(n,M,\epsilon_c,\epsilon_b,\epsilon_s)$ \emph{commitment code} over $\mathcal{G}$ is a two-party interactive protocol between a committer (Alice) and a receiver (Bob) with the following phases and properties.

\textbf{1. Commit phase:}
For every commitment $m\in[M]$, the channel $\mathcal{G}$ is used $n$ times. Between any two uses of the noisy channel, the parties may exchange an arbitrary finite amount of noiseless communication. All parties may use private randomness.
We denote by the random variable $V_b$ the receiver's complete view at the end of the commit phase, including all channel outputs, all noiseless communication, and all of the receiver's private randomness. When the sender honestly commits to $m$, the resulting distribution of $V_b$ is denoted by $P_{V_b|m}$.

\textbf{2. Reveal phase and verification:}
After the commit phase, the sender reveals the committed message $m$ by providing an opening $O$ to the receiver. The receiver applies a verification rule $\mathsf{Acc}(V_b,m,O)\in\{0,1\}$, where $\mathsf{Acc}(V_b,m,O)=1$ means that the receiver accepts $m$ as a valid commitment.

\textbf{a) $\epsilon_c$-concealing code:}
For every pair of $m,m'\in[M]$,
\[
\left\|P_{V_b|m}-P_{V_b|m'}\right\|_{\mathrm{TV}}
\leq \epsilon_c.
\]

\textbf{b) $\epsilon_b$-binding code:}
For every committer strategy sending an input sequence $x^n$ through the noisy channel (not necessarily from the DI code in the dishonest case) and generating the openings $O_m$ and $O_{m'}$ for any pair of messages $m,m'\in[M]$,
\[
\Pr\left\{
    \mathsf{Acc}(V_{b},O_m|x^n)=1
    \ \cap\
    \mathsf{Acc}(V_{b},O_{m'}|x^n)=1\right\}
\leq \epsilon_b.
\]

\textbf{c) $\epsilon_s$-sound code:}
For every commitment $m\in[M]$ and honest behaviour by the sender,
\[
\Pr\left\{\mathsf{Acc}(V_b,O_m\mid m)=1
\right\}
\geq 1-\epsilon_s.
\]
\end{definition}

\begin{remark}
We assume that the participants are \emph{rational}, in the sense that they are interested in completing the commitment task according to its objective. Thus, a participant may deviate from the prescribed protocol whenever doing so may improve their outcome, but not merely to prevent the task from being completed. Without this assumption, commitment cannot be meaningfully guaranteed, since the sender could always provide a wrong opening such that the receiver always discards the commitment, and the receiver could always declare that the commitment has not been accepted in spite of the verification test, trivially defeating the task by refusing to complete it.
\end{remark}
\medskip


The paper is organized as follows. In Section~\ref{sec:strategy_concealment} we present our commitment protocol based on a DI code and the concealing property that stems from it. These parts are common in the honest but curious and dishonest cases. 

In Section~\ref{HBC} we study the honest but curious security setting when the strategy in Section~\ref{sec:strategy_concealment} is used. In Subsection~\ref{ssec:binding_HBC}, we analyse bindingness and, in Subsection~\ref{ssec:soundness_HBC}, soundness. The performance of the resulting commitment code and a lower bound on the honest but curious security setting (using the DI codes proposed in \cite{CDBW:Gauss} as building block) are provided in Subsection~\ref{ssec:DI_to_BC_HBC}

The main results on the (more general) dishonest security setting can be found in Section~\ref{sec:dishonest}. We will need a modification of the projection based codes in \cite{CDBW:EW_simplified,CDBW:Gauss} to show bindingness in Subsection~\ref{ssec:binding_dishonest}, and the soundness of the modified code is studied in Subsection~\ref{ssec:sound_dishonest}. The size of the BC code and the rate properties in the asymptotic regime, leading to the main theorem, are studied in Subsection~\ref{ssec:DI_to_BC_dishonest}. 

Section~\ref{sec:previous_constructions} shows how previous commitment code constructions over AWGN channels were actually (hiddenly) also based on typicality-based (suboptimal) DI codes, showcasing the deep relation between DI and commitment. 

We conclude in Section~\ref{sec:conclusions} with a summary, final remarks and open problems.
\medskip


\section{Proposed BC protocol and its concealment}\label{sec:strategy_concealment}
Assume that sender and receiver are connected by a noisy AWGN channel as well as by a noiseless communication link. They agree on an $(n,N,\lambda_1,\lambda_2)$-DI code over $\cG$. This code maps binary vectors $w^k\in\{0,1\}^k$ into the real-valued $n$-dimensional code words $u_i\in\cX^n$ for $i\in[N]$ and $N=2^k$. The linearithmic rate of the DI code is $\dot R_\text{DI}(n)=k/n\log n$. 
The sender commits to one of $M$ messages, each of which is deterministically assigned a different binary string $a^r\in \{0,1\}^r$. Notice that $M=2^r$ and the linearithmic rate of commitment is $\dot R_\text{BC}(n)=r/n\log n$. 

\begin{tcolorbox}[title=Commitment Protocol]
\textbf{Commit Phase:}
\begin{enumerate}
    \item The sender samples a binary string $w^k$ at random, computes its corresponding DI code word $u_i$, and sends it through the AWGN channel.
    \item The sender picks $f:\!\{0,1\}^k\!\to\!\{0,1\}^r$ uniformly at random from a family of two-universal hash functions, computes $f(w^k)$, and sends $z^r=f(w^k)\oplus a^r$ (with $\oplus$ the bitwise \textsc{xor}) and a description of $f$ through the noiseless channel. Thus committing to the bit string $a^r$.
    \item The receiver obtains a noisy output sequence $y^n\in\cY^n$, a noiseless description of the hashing function used, and the sequence $z^r$. 
\end{enumerate}
  \tcblower
\noindent\textbf{Reveal Phase:}
\begin{enumerate}
    \item The sender shares $w^k$ to the receiver through the noiseless channel.
    \item The receiver computes the DI code word $u_i$ and performs the corresponding identification test $\cD_i$ to the noisy sequence $y^n\in\cY^n$ it had received.
    \item If the output is accepted by the identification test, the receiver computes $a^r=z^r\oplus f(w^k)$ and accepts it as a valid commitment.
\end{enumerate}
\end{tcolorbox}

\subsection{The proposed protocol is concealing}\label{ssec:concelaing}
DI has an intrinsic weak sense of concealment by itself. Since there are many overlaps between the decoding regions (although pairwise small), the output sequence $y^n$ of the channel will land, with high probability, inside the decoding regions corresponding to several messages. Therefore, if the receiver tries to guess the DI code word (thus gaining information about the commitment) before the reveal phase by applying an identification test for every $i\in[N]$, they will get many positive answers (many candidates for the sent sequence). 

This can also be intuitively seen by remembering that the DI codes that we aim to use are much larger (they have a linearithmic scale) than any transmission code (which can just have linearly scaling message length).
By the strong converse results in transmission, one can not reliably decode more than $2^{nC_T}$ messages (with $C_T$ the capacity of the channel used), and therefore, it is not possible to reliably guess which out of the $2^{O(n\log n)}$ messages from a DI code was sent.
\medskip

However, the concealment notion needed for BC is a stronger one: the output distributions corresponding to two different commitments must be similar, see in Definition~\ref{def:commitment}a). Our plan consists in using the leftover hashing lemma (Lemma~\ref{lemma:LHL}) to bound the similarity between output distributions. 

Let $W^k$ be the random variable that uniformly samples $w^k\in\{0,1\}^k$. In the commitment construction above, each of these strings is assigned to a DI code word $x^n$. The random variable representing the output probability distribution of the noisy channel when this $x^n$ is input is defined to be $Y^n$. We start with the following min-Entropy bound from \cite{NBSH:commitment_Gauss}:


\begin{theorem}[Privacy amplification with continuous side information \cite{NBSH:commitment_Gauss}]\label{thm:privacy_ampl_continuous}
Given an AWGN channel $\cG$ with transmission capacity $C_\text{T}(\cG)$ and the commitment construction above, we have, for any $\delta>0$ and large enough $n$,
\[
H_{\min}(W^k|Y^n)\geq 
k-n\left[C_\text{T}(\cG)+\delta\right].
\]
\end{theorem}

Substituting the bound in the theorem above into the leftover hashing lemma produces the following:
\begin{align}
\!\!\!\!\left\|P_{F(W^k),F,Y^n}\!-\!P_{U_r}P_{F,Y^n}\right\|_{\text{TV}}\! 
&\leq \frac{1}{2}2^{-\frac{1}{2}\left(H_{\min}(W^k|Y^n)-r\right)}\nonumber\\
&\leq\frac{1}{2}2^{-\frac{1}{2}\left[k-n\left(C_\text{T}(\cG)+\delta\right)-r\right]}\!.\label{eq:concealment_step1}
\end{align}

Now, since the message $a^r$ is chosen independently of $W^k$ and $F$, the distributions $P_{F(W^k),F,Y^n}$, $P_{U_r}$ and $P_{F,Y^n}$ do not depend on $a^r$. Hence, by the data-processing property of the total variation distance, see Eq.~\eqref{eq:data-processing-TVD}, we have that
\begin{equation}\begin{split}\label{eq:concealment_step2}
   & \left\|P_{F(W^k),F,Y^n}\!-\!P_{U_r}P_{F,Y^n}\right\|_{\text{TV}}\\
    &\qquad\geq\left\|P_{a^r\oplus F(W^k),F,Y^n\mid a^r}\!-\!P_{U_r\mid a^r}P_{F,Y^n\mid a^r}\right\|_{\text{TV}}.
\end{split}\end{equation}
Notice now that $P_{a^r\oplus F(W^k),F,Y^n\mid a^r}$ is exactly the output probability distribution of Bob's view after $a^r$ is sent, that is $P_{V_b\mid a^r}$, and also that $P_{U_r\mid a^r}P_{F,Y^n\mid a^r}=P_{U_r}P_{F,Y^n}$. Repeating the argument for any different commitment $\tilde a^r$ we get, by the triangular inequality, that
\begin{align}
&\!\!\!\left\|P_{V_b|a^r}-P_{V_b|\tilde a^r}\right\|_{\text{TV}}\nonumber\\
&\leq \left\| P_{V_b|a^r}-P_{U_r}P_{F,Y^n}\right\|_{\text{TV}} + \left\|P_{U_r}P_{F,Y^n}-P_{V_b|\tilde a^r}\right\|_{\text{TV}}\nonumber\\
&\stackrel{(b)}{\leq} 2^{-\frac{1}{2}\left[k-n\left(C_\text{T}(\cG)+\delta\right)-r\right]}\nonumber\\
&\stackrel{(c)}{\leq} 2^{-\frac{1}{2}\gamma n\log n},
\end{align}
where in $(b)$ we used Equations \eqref{eq:concealment_step1} and \eqref{eq:concealment_step2}, and $(c)$ is true as long as $\gamma \leq \dot R_\text{DI}(n)-\dot R_\text{BC}(n) -\frac{C_T(\cG)+\delta}{\log n}$. We have shown: 

\begin{theorem}\label{thm:concealment}
For any $\delta,\gamma>0$, the proposed commitment protocol is $\epsilon_c$-concealing with
\(
\epsilon_c=2^{-\frac12\gamma n\log n}
\)
as long as
\[\hspace{40pt}\dot R_\text{BC}(n)\leq \dot R_\text{DI}(n)-\gamma -\frac{C_T(\cG)+\delta}{\log n}.\hspace{40pt}\qed\]
\end{theorem}

We study next the honest but curious and the dishonest security settings. Both these scenarios will follow the proposed protocol and must satisfy the derived concealment condition above. The only difference lies on the properties of the DI code used. Indeed, in the dishonest case, modifications have to be applied to ensure bindingness. 
\medskip


\section{Honest but curious security setting}\label{HBC}
In the honest but curious security setting, the two parties can try to cheat but have to follow the protocol guidelines. This concept was first introduced in cryptography under the name ``semi-honest scenario'' in \cite{GMW:semi-honest}. Ahlswede and Csiszár coined the term ``honest but curious'' in a more information theoretical environment \cite{AC:honest-curious}. This section will also serve us as a warm up for the general results in the next section.

Applied to the commitment protocol proposed in the previous section, an honest but curious sender is forced to input a valid DI code word, say $u_i$, to the noisy channel. This word is tied to a particular binary string $w^k$. They also need to pick a hashing function $f$ and send $z^r=f(w^k)\oplus a^r$ and a description of $f$ through the noiseless channel, thus committing to the message $a^r$. Up to here, the sender is forced to be honest. During the reveal phase, however, the curious sender might try to transmit a different $\tilde w^k$ (related to a different DI code word $u_j\neq u_i$ and therefore revealing a false commitment $\tilde a^r\neq a^r$) through the noiseless channel, hoping that Bob will accept the false $\tilde a^r$ as a valid commitment, thus breaking the bindingness property. Can this cheating strategy succeed?

\subsection{DI codes are binding in the honest but curious setting}\label{ssec:binding_HBC}

Let us take any reliable $(n,N,\lambda_1,\lambda_2)$-DI code over an AWGN channel with code words and their corresponding identifiers (decoders) denoted by $u_i$ and $\cD_i$ for $i\in[N]$, respectively.
We need to calculate, for any pair of commitments $a^r$ and $\tilde a^r$, the probability that $\mathsf{Acc}(V_{b},O_{a^r}|u_i)=1$ and $\mathsf{Acc}(V_{b},O_{\tilde a^r}|u_i)=1$ simultaneously hold. 

The verifiers in the commitment protocol proposed in Section~\ref{sec:strategy_concealment} are exactly the identification tests $\cD_i$. Hence, the probability that the commitment $\tilde a^r$, related through the hashing to the binary string $\tilde w^k$, and hence to the DI word $u_j$ with decoder $\cD_j$, is accepted is given, by the definition of a reliable DI code (see Definition~\ref{def:DI_code}), by 
\begin{equation}
\Pr\{\mathsf{Acc}(V_{b},O_{\tilde a^r}|u_i)\!=\!1\}=\cG^n(D_j\mid u_i)\leq\lambda_2.
\end{equation}
Then, by simply noticing that $\Pr\{A\cap B\}\leq P(B)$, we can conclude that, for any honest input sequence $u_i$ (a DI code word), and for any pair of commitments $a^r,\tilde a^r$, we have
\begin{equation}
\!\!\Pr\left\{
    \mathsf{Acc}(V_{b},O_{a^r}|u_i)\!=\!1
    \, \cap\,
    \mathsf{Acc}(V_{b},O_{\tilde a^r}|u_i)\!=\!1\right\}
\leq \lambda_2.
\end{equation}

Then, a reliable DI code is binding in the honest but curious case. Indeed, we have proven the following.

\begin{theorem}
Given any reliable $(n,N,\lambda_1,\lambda_2)$-DI code, the proposed commitment protocol is $\epsilon_b$-binding, with $\epsilon_b=\lambda_2$, in the honest but curious security picture.\hfill\qed
\end{theorem}

To avoid any confusion, let us expose a key remark: in the previous section we have explained that, due to the abundant overlaps of the identification regions, when a sequence is sent it will likely be accepted by many tests. Why does not this break bindingness automatically? 

Well, let Alice send a code word $u_i$. This word will likely be accepted by the test $\cD_i$ and rejected by any other test $\cD_j$, but the probability that it is rejected by \emph{all} other tests is tiny. So, while Alice knows that her code word will probably be accepted by several DI tests apart from $\cD_i$, she has no way of reliably knowing \emph{which tests} are those. In fact, if she tries to guess, since she does not know the exact output $y^n$, she will likely fail. This argument can be repeated in the case Alice sends an input sequence, not necessarily from the code, but the DI code is binding. That is the dishonest scenario which will be studied in Section~\ref{sec:dishonest}.

\subsection{Soundness of the honest but curious BC codes}\label{ssec:soundness_HBC}

As we have already argued in the introduction, reliable $(n,N,\lambda_1,\lambda_2)$-DI codes are sound for commitment. Indeed, under the proposed strategy, the probability of Bob accepting the true commitment is equal to the probability that the decoder $\cD_i$ correctly identifies the DI code word $u_i$ (the rest of the protocol is noiseless and deterministic). Then,
\begin{equation}
\Pr\left\{\mathsf{Acc}(V_b,O_m\mid m)\!=\!1\right\}
=\cG^n(D_i\mid u_i)
\geq 1-\lambda_1.
\end{equation}
In other words, we have shown the following:

\begin{theorem}\label{thm:sound_HBC}
Given any reliable $(n,N,\lambda_1,\lambda_2)$-DI code, the proposed commitment protocol is $\epsilon_s$-sound, with $\epsilon_s=\lambda_1$, in the honest but curious security picture.\hfill\qed
\end{theorem}
\medskip

\subsection{From DI to honest but curious commitment codes}\label{ssec:DI_to_BC_HBC}
In the previous subsections we have observed that the commitment protocol proposed in Section~\ref{sec:strategy_concealment} is binding and sound if the used deterministic identification code is reliable. This, together with the concealment condition in Subsection~\ref{ssec:concelaing}, gives us the following result.

\begin{theorem}\label{thm:HBC}
From any reliable $(n,N,\lambda_1,\lambda_2)$-DI code for $\cG$ one can construct an $(n,M,\epsilon_c,\epsilon_b,\epsilon_s)$-BC code with 
\[
\epsilon_c=2^{-\frac12 \gamma n \log n},\quad\epsilon_b=\lambda_2,\quad\text{and}\quad\epsilon_s=\lambda_1.
\]
Furthermore, if $\log N$ scales linearithmically, And $\dot R_{\text{DI}}>0$ is an achievable DI rate, then $\dot R_{\text{BC}}=\dot R_{\text{DI}}$ is an achievable commitment rate.
\end{theorem}
\begin{proof}
    It is only left to show the last claim. Directly from the concealment condition in Theorem~\ref{thm:concealment}, it is clear that, for any $\nu>0$, there exist small enough $\gamma>0$ and large enough $n$ such that 
    \begin{equation}
    \dot R_\text{BC}(n)= \dot R_\text{DI}-\nu\leq \dot R_\text{DI}(n)-\gamma -\frac{C_T(\cG)+\delta}{\log n},
    \end{equation}
    since the last term vanishes as $n$ increases. Hence, the concealment condition can be satisfied for any $\nu>0$. Finally, since $\nu$ can also be chosen arbitrarily small (without breaking the previous argument), we have that $\dot R_\text{DI}$ must be an achievable rate for commitment. 
\end{proof}
\begin{corollary}\label{cor:HBC}
The honest but curious commitment capacity over AWGN channels is lower bounded by $\dot C_\text{hBC}(\cG)\geq\frac12$.
\end{corollary}
\begin{proof}
The $L$-layer capacity-achieving DI codes described in \cite{CDBW:Gauss}, with $L>0$ a constant integer, have 
\begin{equation}
\lambda_1=2L\Phi(-\log n)\quad\text{and}\quad\lambda_2=2\Phi(-\log n),
\end{equation}
which vanish as $n\to\infty$ since, by Proposition~\ref{prop:small_projection}, $\Phi(-\log n)=\exp[-\Theta((\log n)^2)]$. Hence, those codes are reliable and can achieve a linearithmic rate of $\dot R_\text{DI}=\frac12$. Then, by Theorem~\ref{thm:HBC}, $\dot R_\text{BC}=\frac12$ is also achievable. The capacity is the supremum of all achievable rates, and it is hence lower bounded by $\dot C_\text{hBC}(\cG)\geq\dot R_\text{BC}=\frac12$.
\end{proof}


\section{Main results: dishonest security setting}\label{sec:dishonest}
In this section we will observe that a modification of the single-layer DI codes proposed in \cite{CDBW:Gauss, CDBW:EW_simplified} can produce, through the strategy in Section~\ref{sec:strategy_concealment} good commitment codes. 
These DI codes that we want to reuse are based on a specific geometric structure of the points together with suitably defined projective decoders.  

\textbf{Encoding:} The DI code words are $N$ elements $\{u_i\}_{i=1}^N$ at a minimum angular separation $\theta$ in the surface of the ball $\cB_{\vec o}(n,\sqrt{nP})$ that defines the input space.

\textbf{Decoding:} To identify the code word $u_i$, with $i\in[N]$, we use the following projection test:
\begin{equation}\label{eq:decoder_layer}
    \cD_i=\{\vec y\in\RR^n:\|\Pi_{\vec u_i}\vec y-\vec u_i\|_2\leq \sigma\log n\}.
\end{equation}
This decoding region is a slab of width $w=2\sigma\log n$ centred at the head of the vector $\vec u_i$ and perpendicular to $\vec u_i$. See Figure~\ref{fig:1layer_no_binding_slabs} for visual intuition.

It was shown in \cite{CDBW:EW_simplified,CDBW:Gauss} that the construction above is reliable for DI if a suitable minimum angle $\theta$ is choses. In particular, choosing $\theta$ such that 
\begin{equation}\label{eq:angle-dense}
    \cos\theta=1-\frac{2\sigma\log n}{\sqrt{nP}},
\end{equation}
which implies that for any two different code words $u_i\neq u_j$
\begin{equation}\label{eq:d-angle-dense}
\|\vec u_i-\Pi_{\vec u_i}\vec u_j\|_2\geq2\sigma\log n,
\end{equation}
one obtains, with the encoding and decoding strategy above, a reliable $(n,N,\lambda,\lambda)$ DI code, with $\lambda=2\Phi(-\log n)$, achieving a linearithmic rate of $\dot R_{\text{DI}}=\frac14$. Meaning that, if a word is sent and tested, it will be accepted with high probability (this will give us the soundness of the BC protocol in Subsection~\ref{ssec:sound_dishonest}). Whereas if the sent word is different from the tested one, it will almost surely not be accepted (for bindingness in the dishonest case, we need a stronger notion than this one, which will force a modification of the DI code, discussed and studied in Subsection~\ref{ssec:binding_dishonest}). In Subsection~\ref{ssec:DI_to_BC_dishonest} we analyse the rates of the modified DI code and the subsequent commitment code.


\subsection{Projection-based DI codes are sound for commitment}\label{ssec:sound_dishonest}
Soundness studies the probability of rejection of the true commitment after honest behaviour by the two parties. Then, we just need to study the probability that a DI code word of a projection based code, say $u_i$, is rejected by its corresponding decoder $\cD_i$ (this is exactly the first type of error of the DI protocol, as we already argued in Subsection~\ref{ssec:soundness_HBC}). All the other parts of the protocol are noiseless and deterministic and, therefore, do not affect the correctness of the protocol after honest behaviour. We have,
\begin{equation}
\Pr\left\{\mathsf{Acc}(V_b,O_{a^r}\mid a^r)=1\right\}
=\Pr\left(\cD_i\mid u_i\right)\geq1-\lambda_1,
\end{equation}
The last inequality following from Definition~\ref{def:DI_code}. Notice also that the first type of error, a missed identification, in a projection-based DI code such as the one described in the beginning of this section does not depend on the chosen minimum angular separation $\theta$ between code words. Indeed, for any $u_i$, with $i\in[N]$, and its corresponding decoder [defined in Equation~\eqref{eq:decoder_layer}], the following is satisfied:
\begin{align}
\Pr\left(\cD_i\big|\,\vec{u}_i\right)&=\Pr\left\{\|\Pi_{\vec u_i}\vec{Y}-\vec{u}_i\|_2\leq\sigma\log n\,\Big|\,\vec{u}_{i}\right\}\nonumber \\
&=\Pr\left\{\|\Pi_{\vec u_i}(\vec{u}_{i}+\sigma\vec{Z})-\vec{u}_{i}\|_2\leq\sigma\log n\right\}\nonumber \\
&=\Pr\left\{\|\Pi_{\vec u_i}\vec{Z}\|_2\leq\log n\right\}\nonumber\\
&=1-2\Phi\left(-\log n\right),\label{eq:layer_error1}
\end{align}
where the second line follows from the effect of an AWGN channel [see Equation~\eqref{eq:channel_effect}]; the third from the fact that the projection $\Pi_{\vec u_i}\vec u_i=\vec u_i$; and in the last inequality we have used Proposition~\ref{prop:small_projection}. In other words, the missed identification error of the code word $u_i$ is, for all $i\in[N]$ and the superscript $c$ denoting the complementary set, 
\begin{equation}
  \!\!\!P_{e,1}(\vec{u}_i)\!:=\Pr\left(\cD_i^{c}\big|\vec{u}_i\right)\!=1-\Pr\left(\cD_i\big|\vec{u}_i\right)\!=2\Phi\left(-\log n\right)\!.
\end{equation}
 
\begin{theorem}\label{thm:epsilon_soundness}
The proposed commitment protocol using the single-layer DI code described in this section is $\epsilon_s$-sound with $\epsilon_s=2\Phi(-\log n)$ in the (general) dishonest setting. \qed
\end{theorem}


\subsection{About bindingness in the dishonest scenario}\label{ssec:binding_dishonest}
In the fully dishonest case, the sender can choose any input sequence (not necessarily a DI code word) to transmit through the noisy channel. This is a picture that deterministic identification does not contemplate. Indeed, one needs to ensure that the used DI code is binding. 
\begin{definition}
    A \emph{binding DI code} has no input sequence (not necessarily a code word) such that the output is reliably accepted by two different tests.
\end{definition}

Given any sequence $x^n$ input to the noisy channel and for any pair of committed messages $m,m'\in[M]$, one has
\begin{align}
&\!\!\Pr\left\{\mathsf{Acc}(V_{b},O_m|x^n)=1
    \, \cap\,
    \mathsf{Acc}(V_{b},O_{m'}|x^n)=1\right\}\label{eq:binding_min}\\
&\leq\min\!\left[\Pr\left\{\mathsf{Acc}(V_{b},O_m|x^n)\!=\!1\!\right\}\!,\Pr\left\{\mathsf{Acc}(V_{b},O_{m'}|x^n)\!=\!1\!\right\}\right]\!.\nonumber
\end{align}
If the DI code is binding, no input sequence $x^n$ can be reliably accepted by more than one test, and then the minimum probability above must vanish, proving bindingness. The problem is thus reduced to finding and studying such binding DI codes.

In this section, we will observe that a modification of the single-layer code for DI over AWGN channels described in the beginning of the section and extracted from in \cite{CDBW:Gauss,CDBW:EW_simplified} is binding. Furthermore, we will observe that the code is still reliable and its size is not changed at first order.

Let us start by observing that, if one chooses the minimum angle between code words in Equation~\eqref{eq:angle-dense}, the one used in \cite{CDBW:Gauss,CDBW:EW_simplified}, input sequences $x^n$ (not from the code) can be found such that they are likely accepted by two tests $\cD_i$ and $\cD_j$ corresponding to different code words $u_i$ and $u_j$, respectively. 

To understand this, remember that the decoders in Equation~\eqref{eq:decoder_layer} are slabs of width $w=2\log n$ centred and perpendicular to $\vec u_i$, see Figure~\ref{fig:1layer_no_binding_slabs}. Since, by Proposition~\ref{prop:small_projection}, the projection of the output onto the direction tested by the decoder is with high probability close to the sent sequence, the test accepts with high probability a sequence inside the decoding slab.

In DI (and in the honest but curious case) we have enough with code words $\vec u_j\neq \vec u_i$ being away from the decoding slab $\cD_i$ to reduce the wrong identification error (bindingness error, in the honest but curious case). Nevertheless, this is not enough to show bindingness in the dishonest picture. Indeed, one could send a sequence inside the overlapping region between different acceptance slabs [see Figure~\ref{fig:1layer_no_binding_slabs}], which would, with high probability, be accepted by the two corresponding tests.

\begin{figure}
    \centering
    \includegraphics[width=\linewidth]{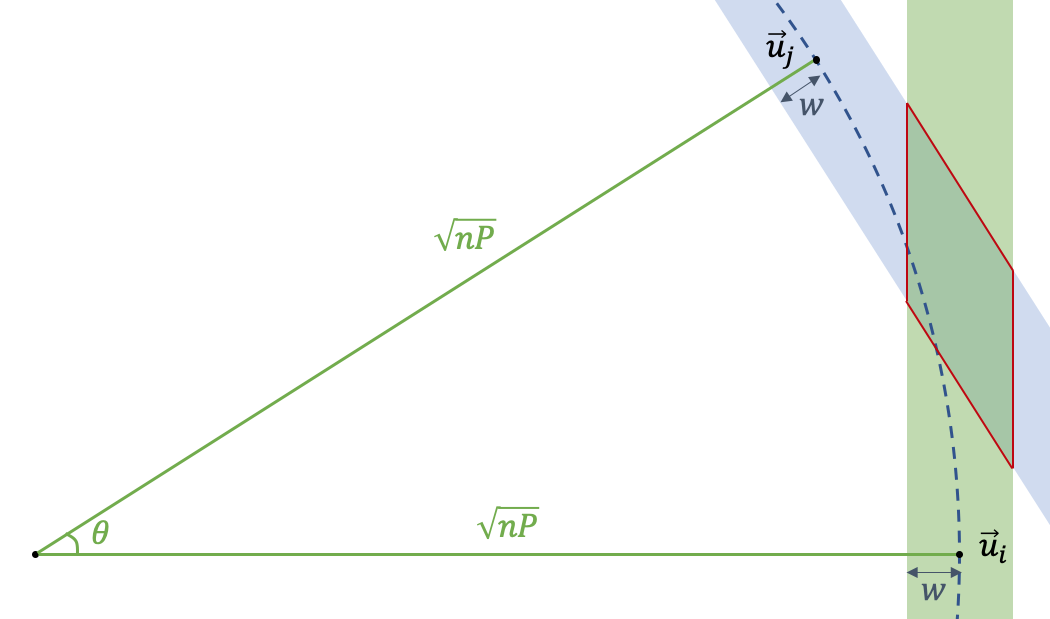}
    \caption{In green the acceptance region of width $w$ of a word $\vec u_i$. In blue, the acceptance region corresponding to $\vec u_j$. If an output sequence falls in the overlap between the two acceptance slabs, both tests are positive.}
    \label{fig:1layer_no_binding_slabs}
\end{figure}

\medskip
Let $u_i$ be the code word assigned to the key sequence $w^k$ and $u_j$ to a different $\tilde w^k$. Then, if there exists a sequence $x^n$ in the decoder overlapping region such that, $\cG^n(\cD_i|x^n)\approx\cG^n(\cD_j|x^n)\geq1-\delta$, for some small $\delta>0$ (the DI code is not binding), the sender can cheat using the following strategy.


\begin{tcolorbox}[title=Cheating strategy using a non-binding DI code]
\noindent\textbf{Commit Phase:}
\begin{enumerate}
    \item The sender inputs $x^n$ to the noisy channel. Then, the receiver obtains a noisy output sequence $y^n$ which will likely be accepted by $\cD_i$ and $\cD_j$.
    \item The sender picks $f$, and sends it together with $z^r=f(w^k)\oplus a^r$ through the noiseless channel, thus committing to $a^r$.
\end{enumerate}

\noindent\textbf{Reveal Phase:}
\begin{enumerate}
    \item[1a)] If the sender honestly shares $w^k$ the receiver computes $u_i$ and performs the test $\cD_i$ to  $y^n$.
    \item[1b)] If the test is positive (which likely happens), the receiver computes $a^r=z^r\oplus f(w^k)$ and accepts it as a valid commitment.
    \item[2a)] If the sender dishonestly shares $\tilde w^k$ the receiver computes $u_j$ and performs the test $\cD_j$ to $y^n$.
    \item[2b)] If the test is positive (which likely happens), the receiver computes $\tilde a^r=z^r\oplus f(\tilde w^k)$ and accepts it as a valid commitment.
    \item[3)] The sender is not bound to its commitment, so it can successfully cheat!   
\end{enumerate}
\end{tcolorbox}
\medskip



There exists, however, a way to overcome the issue. It is based on the fact that all the words of the DI code proposed above live in the surface of the $n$ dimensional sphere $B_{\vec o}(n,\sqrt{nP})$ that defines the available input space (the dashed line in Figures~\ref{fig:1layer_no_binding_slabs} and \ref{fig:1layer_binding_slabs}.

Then, slightly increasing the angles of the initial construction, we can force the overlapping region between code words to separate from the centre (see Figure \ref{fig:1layer_binding_slabs}). In other words, we can increase the power needed to send a non-binding (cheating) code word. The aim of this strategy is to show that there exists a code such that any possible cheating sequence is outside the available input region (it cannot be sent). The needed angle increase does not affect the soundness error (which we already argued was independent of the angle), and we will observe in Section~\ref{ssec:DI_to_BC_dishonest} that it does not change the code size at leading order. 
\medskip
\begin{figure}
    \centering
    \includegraphics[width=0.98\linewidth]{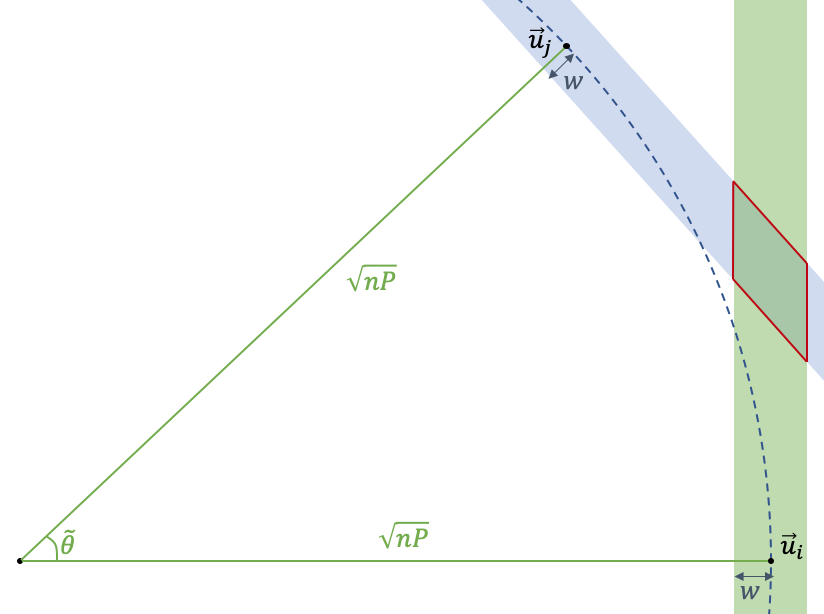}
    \caption{A larger angle pushes the slab overlap to be at least at an extra distance $d$ from the centre of that layer.}
    \label{fig:1layer_binding_slabs}
\end{figure}

Let us calculate the angle $\tilde\theta$ that forces the overlap between any two decoding slabs being separated at least a distance $r+w$, with $r=\sqrt{nP}$ and $w=\sigma\log n$, from the origin. Using basic geometry and Figure \ref{fig:1layer_binding_slabs} for visual intuition, it is clear that the following condition for the minimum angle $\tilde \theta$ between code words must be met: 
\begin{equation}\label{eq:extended_angle}
    \cos \frac{\tilde\theta}{2}=\frac{r-w}{r+w}=1-\frac{2\sigma\log n}{\sqrt{nP}}+O\left(\frac{(\log n)^2}{n}\right).
\end{equation}

In this setting, we show next that there is no input sequence which can be simultaneously accepted by two different identification tests with non-vanishing probability. We start by studying a particular strategy, and then show that there exist no better plan for a dishonest sender. 

Let two code words $\vec u_i$ and $\vec u_j$ be separated by the minimum angular distance $\tilde \theta$. Alice sends the available sequence $x^n$ closest to the overlap between the acceptance slab regions of the two code words at $\ell$. This is the maximum power sequence $\|x^n\|_2=\sqrt{nP}$ at an angular distance $\tilde\theta_{\ell}/2$ from both words. This sequence satisfies the following:
\begin{align}
\|\Pi_{\vec u_i}\, x^n - \vec u_i\|_2
&=\|\Pi_{\vec u_j}\, x^n - \vec u_j\|_2\nonumber\\
&=r- \|x^n\|_2 \cos \frac{\tilde\theta}{2}\nonumber\\
&= \sqrt{nP}\left[1-1+\frac{2\sigma\log n}{\sqrt{nP}}-O\left(\frac{(\log n)^2}{n}\right)\right]\nonumber\\
&\leq2\sigma\log n,\label{eq:opt_strategy}
\end{align}
where the second line follows from basic trigonometry; in the third we substituted the values $\|x^n\|_2=r=\sqrt{nP}$ and $\cos\tilde\theta_\ell/2$ in Equation~\eqref{eq:extended_angle}; and in the fourth line we have simplified the expression giving a lower bound by removing the last negative lower order term.

Now we can calculate the probability that this sequence is accepted by the decoder proposed in Equation~\eqref{eq:decoder_layer}:
\begin{align}
&\Pr\left(\cD_{i}\big|\,x^n\right)=\Pr\left\{\|\Pi_{\vec u_i}\vec{Y}-\vec{u}_{i}\|_2\leq \sigma\log n\, \Big|\,x^n\right\}\nonumber\\
&=\Pr\left\{\|\Pi_{\vec u_i}(x^n+\sigma\vec{Z})-\vec{u}_{i}\|_2\leq \sigma \log n\right\}\nonumber\\
&\leq\Pr\left\{\|\Pi_{\vec u_i}x^n-\vec u_{i}\|_2-\sigma \|\Pi_{\vec u_i}\vec{Z}\|_2 \leq\sigma \log n\!\right\}\nonumber\\
&\leq\Pr\left\{\sigma\|\Pi_{\vec u_i}\vec{Z}\|_2\geq\sigma\log n\right\}\nonumber\\
&=2\Phi\left(-\log n\right), \label{eq:second_error}
\end{align}
where in the second line we have introduced the effect of the AWGN channel [see Equation~\eqref{eq:channel_effect}], in the third we have used the reverse triangle inequality, the fourth follows from the calculation in Equation~\eqref{eq:opt_strategy}, and the last line from Proposition~\ref{prop:small_projection}.
Thus, the probability that Bob will accept the word $u_i$ if $x^n$ is sent vanishes with $n$. An absolutely analogous calculation similarly shows the same for the word $u_j$: 
\begin{equation}
\Pr\left(\cD_{j}\big|\,x^n\right) \leq 2\Phi\left(-\log n\right) = \exp\left[-O\left(\log n\right)^2\right].\label{eq:second_error_bis}
\end{equation}

Therefore, if Alice sends the proposed sequence $x^n$, with high probability, both tests $\cD_i$ and $\cD_j$ will be negative. In other words, If Alice commits to a word $a^r$ by sending the required information with the key $w^k$ related to the DI code word $\vec u_i$, and aims to reveal a different word $\tilde a^r$ by disclosing a key $\tilde w^k$ related to the code word $\vec u_j$, Bob will not accept $\tilde a^r$ as a valid commitment. Because when he receives the key $\tilde w^k$ and hence tests the output of the noisy channel with $\cD_j$ the probability of a positive outcome vanishes [see Equation~\eqref{eq:second_error_bis}].

Actually, even if Alice revealed the correct key $w^k$ during the reveal phase, the commitment would also be rejected with high probability since she sent the cheating sequence $x^n$ instead of a true code word, and the $\Pr\left(\cD_{i}\big|\,x^n\right)$ also vanishes [see Equation~\eqref{eq:second_error}].

Now, notice that if we try to send any other input sequence $\tilde x^n$ through the noisy channel at least one of $\Pr\left(\cD_{u_i}|\,\tilde x^n\right)$ or $\Pr\left(\cD_{u_j}|\,\tilde x^n\right)$ must get smaller. Indeed, let Alice send a maximum power sequence with less projected distance $\|\Pi_{\vec u_i}\,\tilde x^n-\vec u_i\|_2$ to one of the two words, say $u_i$. While she augments the probability that that word can be accepted by the test $\cD_i$ [see the fourth line in Equation~\eqref{eq:second_error}]; she necessarily increments the projected distance $\|\Pi_{\vec u_j}\, x^n - \vec u_j\|_2$ to the other word $u_j$ (therefore reducing the probability that the test $\cD_j$ will accept it). 

Also the projected distance to the code words increases if the input sequence is not of maximum power, hence reducing the probability of acceptance. 

The arguments above allows us to conclude that, for any input sequence $\tilde x^n$ to the noisy channel and any pair of commitments $a^r$ and $\tilde a^r$ (uniquely related through the noiseless communication elements to the DI code words $u_i$ and $u_j$ respectively) satisfies,
\begin{align}
&\min\left[\Pr\left\{\mathsf{Acc}(V_{b},O_{a^r}|\tilde x^n)\!=\!1\!\right\}\!,\Pr\left\{\mathsf{Acc}(V_{b},O_{\tilde a^r}|\tilde x^n)\!=\!1\!\right\}\right]\nonumber\\
&\qquad=\min \left[\Pr\left(\cD_{u_i}|\,\tilde x^n\right), \,\Pr\left(\cD_{u_j}|\,\tilde x^n\right) \right]\nonumber\\
&\qquad\leq2\Phi\left(-\log n\right).
\end{align}
Introducing this into Eq.~\eqref{eq:binding_min} proves the following result:

\begin{theorem}\label{thm:epsilon_binding}
Using the binding DI code above, the proposed BC protocol is $\epsilon_b$-binding with $\epsilon_b=2\Phi\left(-\log n\right)$.\qed
\end{theorem}

\subsection{From a binding DI code to a dishonest commitment code}\label{ssec:DI_to_BC_dishonest}
In Section~\ref{sec:strategy_concealment}, we have shown that the proposed commitment protocol is concealing under a certain rate condition, see Theorem~\ref{thm:concealment}. We have shown soundness in Subsection~\ref{ssec:sound_dishonest} and bindingness in Subsection~\ref{ssec:binding_dishonest}. It is only left to calculate the size of the commitment code and study its rate.

Let us start bounding the size of the modified binding DI code used in the dishonest commitment setting following the arguments in \cite{CDBW:Gauss}. Start by remembering that, the code words are arranged in the surface of the $n$-dimensional ball with radius $r=\sqrt{nP}$ and a minimum angle $\tilde\theta$ between them. This is known as a spherical $\tilde\theta$ packing \cite{packing_book, Wyner:packing_bound, shannon:Gauss, GF:packing_bound}, and the number of points that can be arranged under this condition is lower bounded by \begin{equation}\label{eq:general_packing bound}
N \geq (1+o(1))\frac{\sqrt{2\pi n}\cos\tilde\theta}{\sin^{n-1}\tilde\theta}.
\end{equation}
Since we have from Equation~\eqref{eq:extended_angle} that $\cos(\tilde\theta/2)\geq1-\frac{2w}{r}$, then $\cos(\tilde\theta)\geq1-\frac{8w}{r}$, and $\sin(\tilde \theta)\leq2\sqrt{4w/r}$, with $w=\sigma\log n$ and $r=\sqrt{nP}$. Substituting into Equation~\eqref{eq:general_packing bound}, we can bound
\begin{align}
    \log N&\geq (n-1)\log\frac{1}{\sin\tilde \theta}-O(\log n)\nonumber\\
    &=(n-1)\log \sqrt{\frac{\sqrt{nP}}{8\sigma\log n}} -O(\log n)\nonumber\\
    &=\frac{n\log n}{4}-O(n\log\log n),
\end{align}
where in the first inequality we have taken the logarithm and absorbed the lower order terms into $O(\log n)$, in the second line we introduced the value of $\sin\tilde \theta$, and simplified the expression to leading order in the last line. Hence, the linearithmic DI rate is given by
\begin{equation}
    \dot R_\text{DI}(n):=\frac{\log N}{n\log n}\geq \frac{1}{4}-O\left(\frac{\log\log n}{\log n}\right).
\end{equation}
For large enough $n$ the last term vanishes. Hence, $\dot R_{\text{DI}}=\frac14$ is an achievable rate for the binding DI code. 

Finally, following the last steps in the proof of Theorem~\ref{thm:HBC}, $\dot R_{\text{BC}}=\frac14$ is also an achievable linearithmic rate for dishonest BC. This argument, together with the the concealment, bindingness, and soundness calculations in Sections \ref{ssec:concelaing}, \ref{ssec:binding_dishonest}, and \ref{ssec:sound_dishonest} show the following:

\begin{theorem}\label{thm:main}
Given a binding projection-based DI code and $\gamma$, one can construct a BC code in the dishonest setting, following the strategy in Section~\ref{sec:strategy_concealment}, which is
\[
\begin{split}
&\epsilon_c \text{-concealing with}\hspace{10pt}\epsilon_c=2^{-\frac12\gamma n\log n},\\
&\epsilon_b \text{-binding with}\hspace{23pt}\epsilon_b=2\Phi\left(-\log n\right),\\
&\epsilon_s \text{-sound with}\hspace{29pt}\epsilon_s=2\Phi\left(-\log n\right),
\end{split}
\]
achieving a linearithmic rate of $\dot R_{BC}=\frac14$. Since $\epsilon_c,\epsilon_b,\epsilon_s\to0$ as $n\to\infty$, the code is reliable.
\end{theorem}

\begin{corollary}
    The linearithmic commitment capacity over the AWGN channel $\cG$ is bounded by $\dot C_{BC}(\cG)\geq\frac15$.\hfill\qed
\end{corollary}

\section{Previous BC constructions are based on DI}\label{sec:previous_constructions}
One can observe that the commitment code constructions over AWGN channels in \cite{NBSH:commitment_Gauss, CDN:capacity_UNC} can actually also be traced back to a DI code! Indeed, the code words input to the AWGN channel and the corresponding decoders form what we call a \emph{typicality-based} DI code. These are the first codes that were proposed for DI over AWGN channels \cite{SPBD:DI_power,DI-fading}, and they are based on conditional typicality properties of the input and output sequences of the noisy channel. 

Typicality-based DI codes already achieve a linearithmic scaling of the word length with the block length $n$. However, the achievable rate is limited to $\dot R_\text{DI}^{(\cT)}(\cG)=\frac14$, while the strong converse on the capacity was $\dot C_\text{DI}(\cG)\leq\frac12$. Recently, codes based on projective (rather than typicality) ideas started to close this gap between lower and upper bounds \cite{galaxy-codes,CDBW:EW_simplified}. The DI capacity over AWGN channels was finally established in \cite{CDBW:Gauss} by showing that $\dot R_\text{DI}(\cG)=\frac12$ is an achievable rate when using a code optimising the projective-test ideas (a modification of this capacity-achieving code is the one we have used in this project). A deeper discussion on the similarities and differences on typicality-based and projection-based DI code can be found in \cite[Sec.~I.C]{CDBW:Gauss}.

Now, going back to commitment, and with the insight provided in the previous sections, one can calculate the linearithmic achievable commitment rate of the constructions in \cite{NBSH:commitment_Gauss}. Remember that, there, the capacity was calculated in the linear scale and shown to be infinite. Indeed, they show that the achievable (linear rate) commitment capacity is bounded, for some constants $\gamma,\delta>0$, by 
\begin{equation}\label{eq:old_commitment_rate}
R_\text{BC}\geq R(\alpha)\left[1-\frac{1}{R(\alpha)}(C_T(\cG)-\delta)-\gamma\right],
\end{equation}
with 
\begin{equation}
R(\alpha)\geq-\frac12\log\left[1-\left(1-\frac{\alpha^2}{2}\right)^2\right].
\end{equation}
Please notice a typo on their statement of $R(\alpha)$. Indeed, there, the $\alpha$ is missing a square inside the parenthesis, c.f.~\cite{Gallager:simple_aplication}. Finally, one needs to observe that, due to the structure of the typicality decoders, the parameter $\alpha$ can in fact vanish with $n$ as fast as $\alpha(n)=n^{-1/4}$. Then, 
\(
    R\left(n^{-\frac14}\right)\geq \frac14\log n+O(1/\sqrt{n}),
\)
and
\begin{equation}
    R_\text{BC}(\cG)\geq(1-\gamma)\left[\frac14\log n+O\left(\frac{1}{\sqrt{n}}\right)\right]-C_\text{T}(\cG)+\delta.
\end{equation}
Therefore, the commitment rate calculated on the linearithmic scale is:
\begin{equation}
    \dot R_{\text{BC}}(\cG)=\frac{R_\text{BC}(\cG)}{\log n}\geq\frac{1-\gamma}{4}-O\left(\frac{1}{\log n}\right).
\end{equation}
Hence, since $\gamma>0$ can be chosen arbitrarily small, $\dot R_{\text{BC}}=\frac14$ is an achievable commitment rate in the linearithmic scale that stems from the (hidden) typicality-based commitment code in \cite{NBSH:commitment_Gauss}.

This result strengthens the duality between commitment and DI. In this section, we have shown that commitment codes can achieve a linearithmic rate of $\frac14$ when constructed through typicality-based DI codes, a sort of sub-optimal DI codes which can at most achieve this same rate $\frac14$. 

In this paper, we have decided to show the commitment results using instead a projection-based decoding approach. The reason for this is not only its simplicity, but also the fact that these DI codes can achieve rates of $\dot R_{\text{DI}}=\frac12$. In fact, these are the codes we used to show the lower bound on the honest but curious commitment capacity $\dot C_{\text{hBC}}(\cG)\geq\frac12$. For the dishonest case, we have used a simplified, single-layer projection-based code which is sub-optimal (in the sense that it only achieves a rate of $\dot R_{\text{DI}}=\frac14$, as in the typicality based approach). Because we have not yet found a way to ensure bindingness of multi-layer constructions (needed to achieve the higher rate of $\dot R_{\text{DI}}=\frac12$). If one manages to find multi-layer binding DI codes, the dishonest commitment capacity could be improved to $\dot C_{\text{BC}}(\cG)\geq\frac12$ as in the honest but curious case.


\section{Discussion and conclusions}\label{sec:conclusions}
This work establishes a connection between commitment and deterministic identification. Rather than viewing commitment as an isolated cryptographic primitive, our results show that it can be naturally constructed from DI by imposing an additional bindingness property that prevents a single input sequence from being reliably accepted by multiple identification tests. This perspective provides a new conceptual understanding of BC and suggests that DI captures much of the underlying information-theoretic structure required for secure commitment.

As a first application of this connection, we studied BC over AWGN channels. We observed how, through our construction, the asymptotic properties of DI can be inherited by commitment. Using capacity-achieving DI codes, we obtained commitment schemes operating in the linearithmic regime and established an achievable rate of $\dot R_\text{hBC}=\frac{1}{2}$ in the honest but curious picture, thereby providing a lower bound on the linearithmic commitment capacity in this security setting of $\dot C_\text{hBC}(\cG)\geq\frac{1}{2}$. This result is indeed similar to the linearithmic DI capacity over AWGN channels which is given by $\dot C_\text{DI}(\cG)=\frac12$. Furthermore, in the dishonest case we have constructed commitment codes with half this rate, thus lower bounding the dishonest commitment capacity by $\dot C_{\text{BC}}(\cG)\geq\frac14$.

In spite of the current lack of a converse for BC (which we are eager to study as a next step) the results already inherit one of the most remarkable features of DI. In contrast to classical communication problems, where capacity is determined at first order by channel parameters (the average power constraint and the noise, in fact, its ratio), the achievable linearithmic commitment rate is governed instead by a universal constant. This constant can be traced back to the DI construction and ultimately originates from the geometry of the family of output distributions.

Another interesting property of DI is how the rate scaling is tied to the speed at which the errors vanish. Looking at Theorem~\ref{thm:main}, notice the difference between the behaviour of the bindingness and soundness errors (which are tied to the DI construction) and the concealment error (which is not). The first two vanish sub-exponentially fast in $n$ and, in fact, if one imposes an exponentially vanishing error behaviour the linearithmic scaling of the code is lost \cite{RRGauss-arXiv,CDBW:Reliability-TCOM,CDBW:Gauss}, and therefore so would the BC linearithmic behaviour of our construction. In contrast, the concealing error, which depends on the Hashing construction, vanishes super-exponentially fast!

The efficiency of the protocol is worth remarking. First, on the use of the noisy channel. Indeed, with only $n$ uses of the AWGN link, our strategy can commit reliably to $\sim 2^{\frac n2\log n}$ different messages. This is much more efficient than what one observes for noisy DMCs, where only $\sim 2^{Rn}$ commitments are available. But notice also the efficiency on the use of the noiseless channel. In the general definition of the commitment protocol (Def.~\ref{def:commitment}), the two parties are allowed to share arbitrary finite amounts of noiseless information with every use of the noisy channel. In our construction, however, we only need one-way communication to send $k+\log|\cF|$ noiseless bits during the commit phase (to send $z^r$ and the description of the used hashing function), and $k$ more in the reveal phase (to send $w^k$). We cannot be precise about the exact number of uses of the noiseless channel since we maintained the generality over the family of hashing functions $\cF$. To complete this argument, however, just notice that there exist families of 2-universal hash functions $f\in\cF:\{0,1\}^k\to\{0,1\}^r$ of size $|\cF|=2^{O(k)}$. Hence, using one of these, the amount of noiseless communication throughout the protocol would be of only $\Theta(k)=\Theta(n\log n)$ bits. 

The connection uncovered in this paper appears to be substantially more general than the AWGN setting. While the Gaussian channel provides a particularly clean illustration, we expect the relationship between commitment and deterministic identification to extend to a much broader class of channels. In particular, there are other continuous channels for which linearithmic capacity-achieving DI codes resulting from adaptations of the one in \cite{CDBW:Gauss} (the one that we have used in this paper) are known \cite{CDBW:Bern_arxiv}. These provide a natural next step to study.

Another interesting case is that of discrete memoryless channels. Both BC and DI capacities are already well understood for DMCs. Both these cases operate in the linear regime (instead of the linearithmic behaviour we find for continuous channels), and the code sizes in this regime are asymptotically governed by two different information-theoretic aspects (rather than the geometrical ones we find for continuous channels in the leading order of the linearithmic regime): the maximum entropy of the input in DI, and the equivocation of the channel in BC. These two quantities are generally different (although related), can one construct a BC protocol for DMCs from a DI code and identify the origin of such a difference?


Finally, our framework sheds new light on the previous construction used in \cite{NBSH:commitment_Gauss} to show that the capacity in the linear scale is infinite. By tracing back their steps, one can observe that that construction also relies on a DI code! In fact, it is similar to the initial code constructions for DI over AWGN channels \cite{SPBD:DI_power,DI-fading}. This strengthens our intuition that the connection between DI and commitment might be even more fundamental than the one unveiled by this work.



\bigskip\noindent
{
\textbf{Acknowledgments.} 
The authors would like to thank S. Mamindlapall for inspiring early discussions with C. Deppe and A. Winter on the possible connections between identification and commitment. These discussions motivated the line of research that ultimately led to the present work.

P. Colomer and A. Winter were supported by the Institute for Advanced Study of the Technical University Munich. 

P. Colomer, C. Deppe, and H. Boche acknowledge the support by the German Federal Ministry of Research, Technology and Space (BMFTR) within the National Initiative on 6G Communication Systems through the Research Hub 6G-life under Grants 16KIS2414 and 16KIS2415; and also through the project QR.N under Grants 16KIS2195 and 16KIS2196. H. Boche and C. Deppe were further supported through the project Q-STARS under grants 16KIS2611 and 16KIS2602.

P. Colomer and H. Boche were also supported by the BMFTR within the project CommUnity under Grant 16KISS013, and by the German Research Foundation (DFG) as part of Germany’s Excellence Strategy–EXC 2050/2–Project ID 390696704–Cluster of Excellence “Center for Tactile Internet with Human-in-the-Loop (CeTI)” of Technische Universität Dresden.

A. Winter is supported by the Alexander von Humboldt Foundation; furthermore by the German Research Foundation (DFG) under Germany’s Excellence Strategy: Cluster of Excellence ``Matter and Light for Quantum Computing'' (ML4Q) EXC 2004/2–390534769; and in part by the Spanish MICIN (project PID2022-141283NB-I00) with the support of FEDER funds.
}


\bibliographystyle{ieeetr}
\bibliography{ID}

@ARTICLE{Hayashi:commitment_list-decoding,
  author={Hayashi, Masahito},
  journal={IEEE Transactions on Information Theory}, 
  title={Secure List Decoding and its Application to Bit-String Commitment}, 
  year={2022},
  volume={68},
  number={6},
  pages={3620-3642},
  doi={10.1109/TIT.2022.3150498}}

@INPROCEEDINGS{BJMY:noisy_settings_AWGN,
  author={Budkuley, Amitalok and Joshi, Pranav and Mamindlapally, Manideep and Yadav, Anuj Kumar},
  booktitle={2023 IEEE International Symposium on Information Theory (ISIT)}, 
  title={On the (Im)possibility of Commitment over Gaussian Unfair Noisy Channels}, 
  year={2023},
  volume={},
  number={},
  pages={483-488},
  doi={10.1109/ISIT54713.2023.10206598}}

@ARTICLE{AD:ID_ViaChannels,
  author={Ahlswede, Rudolf and Dueck, Gunter},
  journal={IEEE Transactions on Information Theory}, 
  title={Identification via channels}, 
  year={1989},
  volume={35},
  number={1},
  pages={15-29},
  doi={10.1109/18.42172}
}

@ARTICLE{CDBW:DI_classical,
  author={Colomer, Pau and Deppe, Christian and Boche, Holger and Winter, Andreas},
  journal={IEEE Transactions on Information Theory}, 
  title={Deterministic identification over channels with finite output: a dimensional perspective on superlinear rates}, 
  year={2025},
  volume={71},
  number={5},
  pages={3373-3396},
  doi={10.1109/TIT.2025.3531301}
}

@ARTICLE{Shannon:TheoryCommunication,
  author={Shannon, Claude E.},
  journal={The Bell System Technical Journal}, 
  title={A mathematical theory of communication}, 
  year={1948},
  volume={27},
  number={3{\&}4},
  pages={379-423 {\&} 623–656},
  doi={10.1002/j.1538-7305.1948.tb01338.x}
}

@ARTICLE{AC:DI,
  author={Ahlswede, Rudolf and Cai, Ning},
  journal={IEEE Transactions on Information Theory}, 
  title={Identification without randomization}, 
  year={1999},
  month={11},
  volume={45},
  number={7},
  pages={2636-2642},
  doi={10.1109/18.796419}
}

@ARTICLE{SPBD:DI_power,
  author={Salariseddigh, Mohammad J. and Pereg, Uzi and Boche, Holger and Deppe, Christian},
  journal={IEEE Transactions on Information Theory}, 
  title={{Deterministic Identification Over Channels With Power Constraints}}, 
  year={2022},
  volume={68},
  number={1},
  month={1},
  pages={1-24},
  doi={10.1109/TIT.2021.3122811}
}

@unpublished{RRGauss-arXiv,
    title={Rate-reliability tradeoff for deterministic identification over {Gaussian} channels}, 
    author={Pau Colomer and Christian Deppe and Holger Boche and Andreas Winter},
    year={2026},
    note = {ArXiv[cs.IT]:2602.12182},
    month = {2},
  doi = {10.48550/arXiv.2602.12182}
}

@misc{CDBW:Bern_arxiv,
      title={Deterministic identification for Bernoulli channels and related channels with continuous input}, 
      author={Pau Colomer and Christian Deppe and Holger Boche and Andreas Winter},
      year={2026},
      eprint={2605.05168},
      archivePrefix={arXiv},
      primaryClass={cs.IT},
      url={https://arxiv.org/abs/2605.05168}, 
      note={ArXiv[cs.IT]:2605.05168},
}

@misc{CDBW:Gauss,
      title={Optimal Codes for Deterministic Identification over Gaussian Channels: Closing the Capacity Gap}, 
      author={Pau Colomer and Christian Deppe and Holger Boche and Andreas Winter},
      year={2026},
      eprint={2604.11782},
      archivePrefix={arXiv},
      primaryClass={cs.IT},
      url={https://arxiv.org/abs/2604.11782},
      note={ArXiv[cs.IT]:2604.11782}, 
}

@ARTICLE{CDBW:Reliability-TCOM,
  author={Colomer, Pau and Deppe, Christian and Boche, Holger and Winter, Andreas},
  journal={IEEE Transactions on Communications}, 
  title={Rate-Reliability Tradeoff for Deterministic Identification}, 
  year={2025},
  volume={73},
  number={12},
  pages={14107-14123},
  doi={10.1109/TCOMM.2025.3594790}}

@INPROCEEDINGS{Ja:ID_easier,
  author={JaJa, Joseph},
  booktitle={Proc. 26th Annual Symposium on Foundations of Computer Science (SFCS 1985)}, 
  title={Identification is easier than decoding}, 
  year={1985},
  pages={43-50},
  doi={10.1109/SFCS.1985.32}
}

@article{HanVerdu:ID,
  author={Han, Te Sun and Verd\'u, Sergio},
  title={New results in the theory of identification via channels}, 
  year={1992},
  month={1},
  journal = {IEEE Transactions on Information Theory},
  volume={38},
  number={1},
  pages={14-25},
  doi={10.1109/18.108245}
}

@INPROCEEDINGS{DI-fading,
  author={Salariseddigh, Mohammad J. and Pereg, Uzi and Boche, Holger and Deppe, Christian},
  booktitle={Proc. 2020 IEEE Information Theory Workshop (ITW)}, 
  title={{Deterministic Identification Over Fading Channels}}, 
  year={2021},
  pages={1-5},
  doi={10.1109/ITW46852.2021.9457587}
}

@INPROCEEDINGS{DI-poisson,
  author={Salariseddigh, Mohammad J. and Pereg, Uzi and Boche, Holger and Deppe, Christian and Schober, Robert},
  booktitle={Proc. 2021 IEEE Globecom Workshops}, 
  title={{Deterministic Identification Over Poisson Channels}}, 
  year={2021},
  pages={1-6},
  doi={10.1109/GCWkshps52748.2021.9682110}
}

@ARTICLE{Wyner:packing_bound,
  author={Wyner, A. D.},
  journal={The Bell System Technical Journal}, 
  title={Capabilities of bounded discrepancy decoding}, 
  year={1965},
  volume={44},
  number={6},
  pages={1061-1122},
  doi={10.1002/j.1538-7305.1965.tb04170.x}}

@article{GF:packing_bound,
  author  = {Gil Fernandez, Irene and Kim, Jaehoon and Liu, Hong and Pikhurko, Oleg},
  title   = {New lower bounds on kissing numbers and spherical codes in high dimensions},
  journal = {American Journal of Mathematics},
  volume  = {147},
  number  = {4},
  pages   = {901--925},
  year    = {2025},
  doi     = {10.1353/ajm.2025.a966288},
  issn    = {0002-9327}
}

@ARTICLE{Shannon:Gauss,
  author={Shannon, Claude E.},
  journal={The Bell System Technical Journal}, 
  title={Probability of error for optimal codes in a Gaussian channel}, 
  year={1959},
  volume={38},
  number={3},
  pages={611-656},
  doi={10.1002/j.1538-7305.1959.tb03905.x}}

@INPROCEEDINGS{CDBW:EW_simplified,
  author={Colomer, Pau and Deppe, Christian and Boche, Holger and Winter, Andreas},
  booktitle={European WIRELESS 2026; 31st European Wireless Conference}, 
  title={Simplified Achievability Scheme for Deterministic Identification over Gaussian Channels}, 
  year={2026},
  pages={},
  doi={}
}

@inproceedings{GMW:semi-honest,
author = {Goldreich, O. and Micali, S. and Wigderson, A.},
title = {How to play ANY mental game},
year = {1987},
isbn = {0897912217},
publisher = {Association for Computing Machinery},
address = {New York, NY, USA},
url = {https://doi.org/10.1145/28395.28420},
doi = {10.1145/28395.28420},
booktitle = {Proceedings of the Nineteenth Annual ACM Symposium on Theory of Computing},
pages = {218–229},
numpages = {12},
location = {New York, New York, USA},
series = {STOC '87}
}

@ARTICLE{AC:honest-curious,
  author={Ahlswede, R. and Csiszar, I.},
  journal={IEEE Transactions on Information Theory}, 
  title={Common randomness in information theory and cryptography. I. Secret sharing}, 
  year={1993},
  volume={39},
  number={4},
  pages={1121-1132},
  doi={10.1109/18.243431}}

@ARTICLE{BBCM:privacy_ampl,
  author={Bennett, C.H. and Brassard, G. and Crepeau, C. and Maurer, U.M.},
  journal={IEEE Transactions on Information Theory}, 
  title={Generalized privacy amplification}, 
  year={1995},
  volume={41},
  number={6},
  pages={1915-1923},
  doi={10.1109/18.476316}}

@ARTICLE{Burnashev:ID_gauss,
  author={Burnashev, M.V.},
  journal={IEEE Transactions on Information Theory}, 
  title={On identification capacity of infinite alphabets or continuous-time channels}, 
  year={2000},
  volume={46},
  number={7},
  pages={2407-2414},
  doi={10.1109/18.887853}}

@ARTICLE{Shannon:noise,
  author={Shannon, C.E.},
  journal={Proceedings of the IRE}, 
  title={Communication in the Presence of Noise}, 
  year={1949},
  volume={37},
  number={1},
  pages={10-21},
  doi={10.1109/JRPROC.1949.232969}}

@ARTICLE{AWGN_radio,
  author={Ozarow, L.H. and Shamai, S. and Wyner, A.D.},
  journal={IEEE Transactions on Vehicular Technology}, 
  title={Information theoretic considerations for cellular mobile radio}, 
  year={1994},
  volume={43},
  number={2},
  pages={359-378},
  doi={10.1109/25.293655}}

@ARTICLE{AWGN_IT,
  author={Biglieri, E. and Proakis, J. and Shamai, S.},
  journal={IEEE Transactions on Information Theory}, 
  title={Fading channels: information-theoretic and communications aspects}, 
  year={1998},
  volume={44},
  number={6},
  pages={2619-2692},
  doi={10.1109/18.720551}}

@book{AWGN_wireless,
  author    = {Tse, David and Viswanath, Pramod},
  title     = {Fundamentals of Wireless Communication},
  publisher = {Cambridge University Press},
  year      = {2005},
  address   = {Cambridge},
  isbn      = {9780521845274},
  doi       = {10.1017/CBO9780511807213}
}

@book{packing_book,
  author       = {Conway, John H. and Sloane, Neil J. A.},
  title        = {Sphere Packings, Lattices and Groups},
  edition      = {3},
  series       = {Grundlehren der mathematischen Wissenschaften},
  volume       = {290},
  publisher    = {Springer},
  location     = {New York},
  year         = {1999},
  isbn         = {978-0-387-98585-5},
  ean          = {9781475765687},
  doi          = {10.1007/978-1-4757-6568-7},
  pagetotal    = {706},
  issn         = {0072-7830},
}

@INPROCEEDINGS{galaxy-codes,
  author={Boche, Holger and Deppe, Christian and Mahmoodi, Safieh and Omidi, Golamreza},
  booktitle={2025 IEEE International Symposium on Information Theory (ISIT)}, 
  title={Galaxy Codes: Advancing Achievability for Deterministic Identification via Gaussian Channels}, 
  year={2025},
  note = {ArXiv[cs.IT]:2501.12548},
  doi={10.1109/ISIT63088.2025.11195636}}

@misc{CDN:capacity_UNC,
      title={On the Commitment Capacity of Unfair Noisy Channels}, 
      author={Claude Crépeau, Rafael Dowsley, Anderson C. A. Nascimento},
      year={2019},
      eprint={1905.10921},
      archivePrefix={arXiv},
      primaryClass={cs.CR},
      url={https://arxiv.org/abs/1905.10921}, 
}

@ARTICLE{NBSH:commitment_Gauss,
  author={Nascimento, Anderson C. A. and Barros, JoÃo and Skludarek, Stefan and Imai, Hideki},
  journal={IEEE Transactions on Information Theory}, 
  title={The Commitment Capacity of the Gaussian Channel Is Infinite}, 
  year={2008},
  volume={54},
  number={6},
  pages={2785-2789},
  doi={10.1109/TIT.2008.921686}}

@ARTICLE{Shannon:secrecy,
  author={Shannon, C. E.},
  journal={The Bell System Technical Journal}, 
  title={Communication theory of secrecy systems}, 
  year={1949},
  volume={28},
  number={4},
  pages={656-715},
  doi={10.1002/j.1538-7305.1949.tb00928.x}}

@article{ZCGLT:secure-computation,
title = {Secure Multi-Party Computation: Theory, practice and applications},
journal = {Information Sciences},
volume = {476},
pages = {357-372},
year = {2019},
issn = {0020-0255},
doi = {https://doi.org/10.1016/j.ins.2018.10.024},
url = {https://www.sciencedirect.com/science/article/pii/S0020025518308338},
author = {Chuan Zhao and Shengnan Zhao and Minghao Zhao and Zhenxiang Chen and Chong-Zhi Gao and Hongwei Li and Yu-an Tan},
}

@INPROCEEDINGS{Y:secure-computation,
  author={Yao, Andrew C.},
  booktitle={23rd Annual Symposium on Foundations of Computer Science (sfcs 1982)}, 
  title={Protocols for secure computations}, 
  year={1982},
  volume={},
  number={},
  pages={160-164},
  doi={10.1109/SFCS.1982.38}}

@InProceedings{EGL:contract-signing,
author="Even, S.
and Goldreich, O.
and Lempel, A.",
editor="Chaum, David
and Rivest, Ronald L.
and Sherman, Alan T.",
title="A Randomized Protocol for Signing Contracts",
booktitle="Advances in Cryptology",
year="1983",
publisher="Springer US",
address="Boston, MA",
pages="205--210",
isbn="978-1-4757-0602-4"
}

@InProceedings{Crepeau:crypto_protocols,
author="Cr{\'e}peau, Claude",
editor="Fumy, Walter",
title="Efficient Cryptographic Protocols Based on Noisy Channels",
booktitle="Advances in Cryptology --- EUROCRYPT '97",
year="1997",
publisher="Springer Berlin Heidelberg",
address="Berlin, Heidelberg",
pages="306--317",
isbn="978-3-540-69053-5"
}

@InProceedings{DKS:OT_and_BC,
author="Damg{\aa}rd, Ivan
and Kilian, Joe
and Salvail, Louis",
editor="Stern, Jacques",
title="On the (Im)possibility of Basing Oblivious Transfer and Bit Commitment on Weakened Security Assumptions",
booktitle="Advances in Cryptology --- EUROCRYPT '99",
year="1999",
publisher="Springer Berlin Heidelberg",
address="Berlin, Heidelberg",
pages="56--73",
isbn="978-3-540-48910-8"
}

@inproceedings{Blum:coin-flipping,
  title     = {Coin Flipping by Telephone},
  author    = {Blum, Manuel},
  booktitle = {Advances in Cryptology: A Report on CRYPTO 81},
  pages     = {11--15},
  year      = {1981},
  editor    = {Allen Gersho},
  publisher = {IEEE Workshop on Communications Security / U.C. Santa Barbara}
}

@article{BC:zero-knowledge,
title = {Minimum disclosure proofs of knowledge},
journal = {Journal of Computer and System Sciences},
volume = {37},
number = {2},
pages = {156-189},
year = {1988},
issn = {0022-0000},
doi = {https://doi.org/10.1016/0022-0000(88)90005-0},
url = {https://www.sciencedirect.com/science/article/pii/0022000088900050},
author = {Gilles Brassard and David Chaum and Claude Crépeau},
}

@article{Naor:BC_pseudorandomness,
title = {Bit commitment using pseudorandomness},
journal = {Journal of Cryptology},
volume = {4},
number = {2},
pages = {151-158},
year = {1991},
issn = {0022-0000},
doi = {https://doi.org/10.1007/BF00196774},
author = {Naor, Moni},
}

@InProceedings{HM:BC_computationally_secure,
author="Halevi, Shai
and Micali, Silvio",
editor="Koblitz, Neal",
title="Practical and Provably-Secure Commitment Schemes from Collision-Free Hashing",
booktitle="Advances in Cryptology --- CRYPTO '96",
year="1996",
publisher="Springer Berlin Heidelberg",
address="Berlin, Heidelberg",
pages="201--215",
isbn="978-3-540-68697-2"
}

@article{Mayers:no-go,
  title = {Unconditionally Secure Quantum Bit Commitment is Impossible},
  author = {Mayers, Dominic},
  journal = {Phys. Rev. Lett.},
  volume = {78},
  issue = {17},
  pages = {3414--3417},
  numpages = {0},
  year = {1997},
  month = {Apr},
  publisher = {American Physical Society},
  doi = {10.1103/PhysRevLett.78.3414},
  url = {https://link.aps.org/doi/10.1103/PhysRevLett.78.3414}
}

@article{LC:no-go,
  title = {Is Quantum Bit Commitment Really Possible?},
  author = {Lo, Hoi-Kwong and Chau, H. F.},
  journal = {Phys. Rev. Lett.},
  volume = {78},
  issue = {17},
  pages = {3410--3413},
  numpages = {0},
  year = {1997},
  month = {Apr},
  publisher = {American Physical Society},
  doi = {10.1103/PhysRevLett.78.3410},
  url = {https://link.aps.org/doi/10.1103/PhysRevLett.78.3410}
}

@article{Kent:relativistic,
  title = {Unconditionally Secure Bit Commitment},
  author = {Kent, Adrian},
  journal = {Phys. Rev. Lett.},
  volume = {83},
  issue = {7},
  pages = {1447--1450},
  numpages = {0},
  year = {1999},
  month = {Aug},
  publisher = {American Physical Society},
  doi = {10.1103/PhysRevLett.83.1447},
  url = {https://link.aps.org/doi/10.1103/PhysRevLett.83.1447}
}

@ARTICLE{Wyner:wiretap,
  author={Wyner, A. D.},
  journal={The Bell System Technical Journal}, 
  title={The wire-tap channel}, 
  year={1975},
  volume={54},
  number={8},
  pages={1355-1387},
  doi={10.1002/j.1538-7305.1975.tb02040.x}}

@InProceedings{WNI:BC_capacity_DMC,
author="Winter, Andreas
and Nascimento, Anderson C. A.
and Imai, Hideki",
editor="Paterson, Kenneth G.",
title="Commitment Capacity of Discrete Memoryless Channels",
booktitle="Cryptography and Coding",
year="2003",
publisher="Springer Berlin Heidelberg",
address="Berlin, Heidelberg",
pages="35--51",
isbn="978-3-540-40974-8"
}

@ARTICLE{Gallager:simple_aplication,
  author={Gallager, R.},
  journal={IEEE Transactions on Information Theory}, 
  title={A simple derivation of the coding theorem and some applications}, 
  year={1965},
  volume={11},
  number={1},
  pages={3-18},
  doi={10.1109/TIT.1965.1053730}}

@inproceedings{LHL_original,
author = {Impagliazzo, R. and Levin, L. A. and Luby, M.},
title = {Pseudo-random generation from one-way functions},
year = {1989},
isbn = {0897913078},
publisher = {Association for Computing Machinery},
address = {New York, NY, USA},
url = {https://doi.org/10.1145/73007.73009},
doi = {10.1145/73007.73009},
booktitle = {Proceedings of the Twenty-First Annual ACM Symposium on Theory of Computing},
pages = {12–24},
numpages = {13},
location = {Seattle, Washington, USA},
series = {STOC '89}
}

@INPROCEEDINGS{IMNW:efficient_protocol,
  author={Imai, Hideki and Morozov, Kirill and A. Nascimento, Anderson C. and Winter, Andreas},
  booktitle={2006 IEEE International Symposium on Information Theory}, 
  title={Efficient Protocols Achieving the Commitment Capacity of Noisy Correlations}, 
  year={2006},
  volume={},
  number={},
  pages={1432-1436},
  doi={10.1109/ISIT.2006.262083}}

@INPROCEEDINGS{MYMB:cost-constraint,
  author={Mamindlapally, Manideep and Yadav, Anuj Kumar and Mishra, Manoj and Budkuley, Amitalok J.},
  booktitle={2021 IEEE International Symposium on Information Theory (ISIT)}, 
  title={Commitment Capacity under Cost Constraints}, 
  year={2021},
  volume={},
  number={},
  pages={3208-3213},
  doi={10.1109/ISIT45174.2021.9518204}}

\end{document}